\documentclass[conference,compsoc]{IEEEtran}

\usepackage[T1]{fontenc}
\usepackage[utf8]{inputenc}
\usepackage{cite}          
\usepackage{graphicx}
\usepackage{tikz}
\usetikzlibrary{positioning,decorations.pathreplacing,arrows.meta}
\usepackage{pifont}       
\usepackage{booktabs}
\usepackage{multirow}
\usepackage{xcolor}
\usepackage{amsmath,amssymb,amsthm}
\usepackage{mathtools}
\usepackage{algorithm}
\usepackage{algorithmic}
\usepackage{microtype}
\usepackage{xspace}
\usepackage{url}
\usepackage{hyperref}
\hypersetup{
	pdfauthor={Homayoun Maleki, Nekane Sainz, Jon Legarda, Igor Santos-Grueiro},
	pdftitle={Sustained Participation as a Security Resource: The Bounded Participation Channel},
	pdfsubject={Computer Science - Cryptography and Security},
	pdfkeywords={Sybil resistance, proof of personhood, participation verification},
	colorlinks=true,
	linkcolor=blue!50!black,
	citecolor=blue!50!black,
	urlcolor=blue!50!black
}

\newtheorem{theorem}{Theorem}
\newtheorem{lemma}{Lemma}
\newtheorem{corollary}{Corollary}
\newtheorem{definition}{Definition}
\newtheorem{assumption}{Assumption}
\newtheorem{proposition}{Proposition}
\theoremstyle{remark}
\newtheorem{remark}{Remark}
\theoremstyle{plain}

\newcommand{\BPC}{BPC\xspace}

\newcommand{\blfootnote}[1]{%
	\begingroup
	\renewcommand{\thefootnote}{}\footnote{#1}%
	\addtocounter{footnote}{-1}%
	\endgroup
}

\title{Sustained Participation as a Security Resource:\\
	The Bounded Participation Channel}

\author{
	\IEEEauthorblockN{Homayoun Maleki}
	\IEEEauthorblockA{\footnotesize DeustoTech, University of Deusto\\
		Bilbao, Spain\\
		h.maleki@deusto.es}
	\and
	\IEEEauthorblockN{Nekane Sainz}
	\IEEEauthorblockA{\footnotesize DeustoTech, University of Deusto\\
		Bilbao, Spain\\
		nekane.sainz@deusto.es}
	\and
	\IEEEauthorblockN{Jon Legarda}
	\IEEEauthorblockA{\footnotesize DeustoTech, University of Deusto\\
		Bilbao, Spain\\
		jlegarda@deusto.es}
	\and
	\IEEEauthorblockN{Igor Santos-Grueiro}
	\IEEEauthorblockA{\footnotesize International University of La Rioja (UNIR)\\
		Logro\~{n}o, Spain\\
		igor.santosgrueiro@unir.net}
}

\begin{document}
	
	\maketitle
	\blfootnote{A substantially different preliminary version of this work,
		under the earlier design name Human Challenge Oracle, was
		previously posted as a non-peer-reviewed preprint by a subset of
		the authors~\cite{maleki2026hco}.}

	\begin{abstract}
		
		Can sustained, per-identity participation be engineered into a
		security resource? Doing so requires verifying it---window by
		window---so that no adversary can amortize it. Most anti-Sybil
		defenses answer a different question: they price identity
		\emph{creation}, not identity \emph{survival}. Once admitted, an adversary
		sustains thousands of accounts for free---CAPTCHAs and Proof-of-Personhood
		verify only at the door, and resource-based defenses such as compute or
		capital amortize across identities instead of pricing them individually.
		
		We introduce the \emph{Bounded Participation Channel} (\BPC), a formal primitive
		that closes this gap. \BPC\ repeatedly issues fresh, identity-bound
		challenges under a strict deadline, enforced by four structural
		properties---identity binding, freshness, real-time response, and bounded
		per-channel throughput---that together yield a provable cost theorem:
		sustaining $s$ identities over $T$ windows costs
		$C(s,T) \geq sT/\tau_h$. The guarantee is solver-agnostic, holding
		regardless of whether a channel is operated by a human, an AI system, or
		a hybrid.
		
		We give a hash-based construction with a publicly verifiable procedure
		under standard cryptographic assumptions, characterize four admissible
		challenge families, and evaluate two against three frontier
		models---GPT-4o, Gemini~2.5~Flash, and Claude~Sonnet~4.5---across 600
		trials. Despite near-perfect accuracy (97--100\%), every model remains
		throughput-bounded, exhibiting single-digit $\tau_h$ from either latency
		or correctness limits. The finding is simple but consequential:
		\emph{solvability does not imply unlimited throughput}---so
		sustained participation can be verified, window by window, and
		thereby engineered into a security resource with a provable,
		linear cost floor, independent of who or what supplies it: unlike
		resources that can be pooled, stockpiled, or transferred, this one
		must be re-earned by every identity, in every window, for as long
		as the system relies on it. Verified this way, sustained
		participation becomes a security resource any system can rely on,
		wherever the real question is not who joined, but who is still
		showing up.
		
	\end{abstract}

	\section{Introduction}
	\label{sec:intro}
	
	Open systems increasingly depend on large populations of
	pseudonymous participants. Online communities, collaborative
	platforms, governance mechanisms, gaming ecosystems, social
	platforms, blockchain-based systems, and anti-abuse infrastructures
	must all confront the same fundamental problem: a single adversary can
	cheaply create and coordinate many identities. The challenge,
	however, is not merely preventing identity creation. It is ensuring
	that maintaining many active identities incurs cost that grows
	proportionally with both the number of identities and the duration
	for which they remain active.
	
	\noindent\textbf{This problem is intensifying, not receding.}
	Automated programs generated a record 53\% of global web traffic in
	2025, with AI-driven bot attacks up more than twelvefold year over
	year~\cite{thales2026badbot}. As autonomous AI agents become
	legitimate participants in open systems---not just adversaries
	impersonating humans---the operative question shifts. It is no
	longer only ``is this participant a human or a bot?'' but
	increasingly ``is this one independent agent, or one process posing
	as a thousand?'' A mechanism answering the second question must be
	agnostic to what operates a given identity---a person, a script, or
	a language model---by design, not as an afterthought bolted onto a
	human-detection scheme. This is the setting the rest of the paper
	targets.
	
	This distinction is fundamental, and most existing defenses get it
	backwards. They treat Sybil resistance as an admission problem:
	verify an identity once, then assume continued participation
	remains meaningful. But an adversary does not need cheap
	identities---it needs identities that stay cheap to keep alive.
	What open systems actually require is a cost structure in which
	maintaining $s$ active identities over $T$ time windows necessarily
	incurs cost
	\[
	C(s,T)=\Omega(sT),
	\]
	growing jointly with both identity count and participation horizon.
	
	\noindent\textbf{Why existing approaches fall short.}
	CAPTCHAs~\cite{vonahn2003captcha} exploit perceptual asymmetries at
	admission time but are fundamentally one-shot: once an account is
	created, continued participation incurs no proportional recurring
	cost. Proof-of-Personhood systems~\cite{borge2017proofofpersonhood,%
		ford2020pseudonym} raise the cost of the ceremony but provide no
	mechanism for continuously pricing participation afterward.
	Social-graph and behavioral
	defenses~\cite{yu2006sybilguard,yu2008sybillimit,danezis2009sybilinfer,%
		wang2013clickstream,frank2013touchalytics,serwadda2013kids} address
	related aspects of abuse but establish no explicit lower bound on
	the cost of maintaining many identities over time. The limitation
	is structural: none of these approaches ties continued
	participation to a resource whose recurring cost is guaranteed
	rather than merely possible. The relevant question is therefore not
	merely whether a resource is scarce or reusable, but whether
	continued participation requires fresh resource expenditure in each
	time window~\cite{maleki2026geometry}. Sustained, per-identity
	participation is a promising candidate precisely because---unlike
	compute or capital---it cannot be pooled or transferred between
	identities~\cite{maleki2026geometry}. This raises the question this
	paper answers:
	
	\begin{quote}
		\emph{Can sustained, per-identity participation be engineered
			into a security resource---one that resists amortization across
			identities, across time windows, and across whatever operates
			the channel?}
	\end{quote}
	
	This paper answers it in general form, independent of
	any one downstream use. We introduce the \emph{Bounded
		Participation Channel} (\BPC), a formal primitive for sustained
	participation verification that any such downstream mechanism---consensus,
	governance, reputation, or otherwise---can build on. Rather than
	verifying participants only at admission time, \BPC\ repeatedly
	requires each active identity to demonstrate fresh participation
	through time-local challenges, cryptographically bound to both
	identity and time window, that must be answered within a strict
	response deadline.
	
	\noindent\textbf{Continuous verification, not accumulated
		reputation.}
	\BPC\ is continuous in a precise sense: it produces, for every
	identity $v$ and window $d$, a verifiable participation bit
	$\pi(v,d) \in \{0,1\}$, where $\pi(v,d)=1$ iff $v$ performed the
	qualifying action for window $d$. Sustained participation over $T$ windows is
	the sequence $\pi(v,1),\ldots,\pi(v,T)$. Crucially, \BPC\ does not
	accumulate these bits into a score, reputation, or any other
	persistent state---each $\pi(v,d)$ is independently re-established,
	and nothing about window $d$ affects the cost of window $d+1$.
	\BPC\ exports a sequence of time-indexed participation proofs; it is
	deliberately silent on what higher-level systems do with them. This
	separation is what keeps \BPC\ a primitive rather than a policy:
	gaming platforms may use $\pi(v,\cdot)$ for AFK detection, governance
	systems for voting eligibility, and reputation or consensus systems
	for scoring participation history---\BPC\ remains agnostic to how
	any of these interpret its output.
	
	The security of \BPC\ rests on four structural properties: identity
	binding, freshness, real-time response, and bounded per-channel
	throughput. Together, they rule out every form of
	amortization---no solution can be shared across identities, replayed
	across windows, or produced faster than a channel's physical
	throughput allows---and induce a provable cost theorem:
	\[
	C(s,T)\geq \frac{sT}{\tau_h},
	\]
	where $\tau_h$ denotes the maximum number of valid responses any
	single channel can produce per window. The guarantee holds
	regardless of solver identity: it depends only on the physical fact
	that no channel can serve unbounded throughput within a deadline,
	not on humans outperforming machines.
	
	To evaluate this claim, we instantiate \BPC\ across four challenge
	families and experimentally evaluate two of them against three
	independently developed frontier AI systems: GPT-4o,
	Gemini~2.5~Flash, and Claude~Sonnet~4.5, across 600 trials. Despite
	near-perfect success rates on perceptual tasks, all models remain
	throughput-bounded, yielding small single-digit values of $\tau_h$:
	\emph{solvability does not imply unlimited throughput}.
	
	\noindent\textbf{Contributions.}
	\begin{itemize}
		\item A formal primitive, \BPC, specified abstractly as an
		Event/Claim/Proof pattern with two operations,
		$\textsf{Prove}$ and $\textsf{Verify}$
		(Section~\ref{sec:spec}), separating what is guaranteed from how
		it is realized.
		\item A layered security definition that distinguishes
		objective-level properties (H1--H2, each mapped to a security
		game), a primitive-level qualitative timing requirement (H3),
		and a separate channel-capacity assumption (H4) needed only for
		the multi-identity cost theorem---so that no single proof's
		validity depends on an assumption about channel throughput.
		\item A cost theorem showing that, under Assumption~\ref{asm:throughput},
		these properties jointly force $C(s,T) = \Omega(sT)$, with
		participation exported as a sequence of time-indexed proofs
		rather than an accumulated score.
		\item A hash-based realization of $\textsf{Prove}$ and
		$\textsf{Verify}$ via an interactive challenge-response exchange,
		with a publicly verifiable procedure under standard cryptographic
		assumptions, and four admissible challenge families with
		browser-based instantiations.
		\item A 600-trial evaluation against three frontier AI systems
		showing that near-perfect accuracy does not defeat the bound.
	\end{itemize}
	
	\section{Related Work}
	\label{sec:related}
	
	\noindent\textbf{Sybil attacks and graph-based defenses.}
	Douceur~\cite{douceur2002sybil} established the fundamental problem:
	in open systems without a central authority, identity replication
	cannot be prevented without binding participation to a scarce resource.
	Graph-based defenses---SybilGuard~\cite{yu2006sybilguard},
	SybilLimit~\cite{yu2008sybillimit},
	SybilInfer~\cite{danezis2009sybilinfer},
	SybilRank~\cite{cao2012syrank}---exploit social network structure to
	bound the number of adversarial identities a community-trust graph can
	tolerate. Clickstream analysis~\cite{wang2013clickstream} and network
	topology~\cite{viswanath2010analysis} provide complementary signals.
	These mechanisms are valuable but orthogonal to \BPC: they constrain
	how many Sybil identities exist, not how much it costs to \emph{sustain}
	each one over time. They provide no per-identity per-window cost bound;
	an adversary that passes the graph filter at admission can sustain many
	identities indefinitely at negligible recurring cost.
	
	\noindent\textbf{CAPTCHA and one-time verification.}
	CAPTCHAs~\cite{vonahn2003captcha} exploit the human--machine perceptual
	gap at account creation. reCAPTCHA~\cite{recaptcha2019} and
	hCaptcha~\cite{hcaptcha2020} add behavioral risk signals.
	These defenses share a critical structural weakness: they are
	\emph{one-shot}. Once an account is created, sustaining it incurs zero
	marginal cost; the adversary is never asked to prove renewed
	participation. This is not an engineering failure---it is a
	definitional one. CAPTCHA does not attempt to bound $C(s,T)$; it
	bounds $C(s,1)$ at creation time only. Meanwhile, the perceptual gap
	CAPTCHAs rely on is eroding: deep learning breaks text
	CAPTCHAs~\cite{bursztein2011text,ye2018yet} and semantic image
	challenges~\cite{sivakorn2016iam,goodfellow2014multidigit} at high
	rates, and solving services commoditize the
	bypass~\cite{motoyama2010recaptchas}. Token-binding
	extensions~\cite{popov2019token} attempt to tie authentication tokens
	to specific devices but do not provide fresh, per-window,
	per-identity challenges and offer no composable cost theorem.
	
	\noindent\textbf{Proof-of-Personhood.}
	PoP systems~\cite{borge2017proofofpersonhood,ford2020pseudonym} enforce
	a one-person--one-identity mapping through ceremonies, biometrics, or
	social vouching. Deployed systems include
	Worldcoin~\cite{worldcoin2023} (iris scanning),
	BrightID~\cite{brightid2022} (social verification), and
	Gitcoin Passport~\cite{gitcoin2023} (credential aggregation).
	Zk-proof-based approaches~\cite{semaphore2023,unirep2023} provide
	privacy-preserving personhood attestations without revealing identity.
	PoP raises the cost of \emph{creating} a verified identity but does
	not constrain ongoing participation: once admitted, an adversary
	sustaining $s$ identities faces no proportional recurring cost.
	The adversary simply acquires $s$ PoP credentials once---possibly
	through credential markets~\cite{golle2004economical}---and sustains
	them indefinitely. PoP achieves $C(s,1) = \Omega(s)$ at creation
	time; \BPC\ targets $C(s,T) = \Omega(sT)$ across the full
	participation lifetime.
	
	\noindent\textbf{Continuous and behavioral authentication.}
	Behavioral biometrics use ongoing interaction signals---
	touchscreen dynamics~\cite{frank2013touchalytics}, keystroke
	timing~\cite{serwadda2013kids}, mouse
	movement~\cite{shen2013continuous}---as persistent authentication
	signals. Risk-based authentication systems~\cite{freeman2016onion}
	combine behavioral and contextual signals to detect account takeover.
	These mechanisms come closest to \BPC\ in spirit---continuous
	rather than one-shot---but they stop short in two ways. They are
	passive: the system infers presence from behavior rather than
	issuing identity-specific, time-bound challenges, so a sophisticated
	adversary can replay or synthesize signals without proportional
	effort per identity. And they produce no exportable, verifiable
	record: there is no analogue of \BPC's per-window participation
	proof $\pi(v,d)$, only an internal risk score computed and consumed
	by a single proprietary provider. \BPC\ instead produces a public,
	cryptographically verifiable proof for every window, usable by any
	downstream system in an open, adversarial setting.
	
	\noindent\textbf{Trusted hardware and rate-bounded attestation.}
	Trusted Execution Environments (TEEs)---Intel
	SGX~\cite{costan2016sgx}, ARM TrustZone---bind computation to
	attested hardware, enabling per-device rate enforcement that is
	robust against software-level attacks. TPM-based attestation and
	secure enclaves can enforce~(H4) by construction, since scaling to
	$s$ identities requires $s$ independently attested devices.
	Privacy Pass~\cite{davidson2023privacypass} uses blind signatures to
	decouple rate-bounded token issuance from identity, enabling anonymous
	but rate-limited participation tokens. These approaches are
	complementary to \BPC: trusted hardware is one concrete instantiation
	of the channel serialization basis for~(H4)
	(Section~\ref{sec:spec}). However, \BPC\ does not
	\emph{require} trusted hardware---it is designed to function with
	software-only channel enforcement (session serialization, deadline
	enforcement) when TEEs are unavailable or undesirable.
	
	\noindent\textbf{Continuous verification of a single resource.}
	A structurally related line of work formalizes periodic
	challenge-response verification that binds time to continued
	possession of a resource. Proofs of
	Retrievability~\cite{juels2007pors,shacham2008compact} and Proof of
	Storage-Time~\cite{ateniese2020post} let a server periodically prove
	to a verifier that it still possesses an outsourced file, without
	re-transmitting it---structurally the closest primitive-level
	precedent to \BPC's per-window proof: both are challenge-response
	protocols whose output is a fresh, per-interval, publicly checkable
	claim. The relationship they formalize, however, runs between one
	prover and one file. There is no notion, built into the primitive
	itself, of binding many independently claimed identities to a joint
	per-channel throughput ceiling; deployments that need Sybil-resistance
	for storage providers add it as a separate economic mechanism (e.g.,
	requiring a distinct replication proof per claimed copy), not as a
	consequence of the storage-time proof's own contract. \BPC\ instead
	makes the per-channel throughput bound part of the primitive's
	contract from the start (Property~H4, Section~\ref{sec:spec}), so
	that the cost of sustaining $s$ identities follows directly from
	Properties~(H1)--(H4) rather than from an auxiliary mechanism layered
	on top.
	
	\noindent\textbf{Resource-based anti-Sybil mechanisms.}
	Proof-of-Work~\cite{nakamoto2008bitcoin} and
	Proof-of-Stake~\cite{kiayias2017ouroboros} bind influence to
	computational effort or economic stake.
	Bonneau et al.~\cite{bonneau2015sok} survey the design space.
	These resources are \emph{parallelizable}: once acquired, hashpower
	or stake can in principle be subdivided across $s$ identities without
	a fresh acquisition event per identity. Resource reusability or
	parallelizability alone does not determine whether sustaining that
	influence over time incurs sublinear or linear economic cost; the
	latter depends on a separate, explicit economic assumption about the
	cost of sustaining the resource, not on reusability
	alone~\cite{maleki2026geometry}. Proof-of-Space and
	storage-based mechanisms~\cite{dziembowski2018spacemint} offer lower
	energy cost but retain the same parallelizability structure.
	Rate-limited token systems such as PoW-based email
	stamps~\cite{dwork1993pricing} impose per-request costs but do not
	bind costs to specific identities across time windows, offering no
	equivalent of $C(s,T) = \Omega(sT)$.
	
	\noindent\textbf{Solver throughput under real-time constraints.}
	\BPC's security argument is solver-agnostic: it does not depend on
	humans outperforming AI or on AI being unable to solve individual
	challenges. The relevant quantity is the \emph{rate} at which valid
	responses can be produced per channel under strict deadlines, not
	per-challenge accuracy. Prior work has documented human and AI
	performance characteristics on perceptual and reasoning
	tasks~\cite{geirhos2018generalisation,geirhos2019texture,%
		geirhos2020shortcut,lake2017building,schutt2023humanvsai}, as well
	as the increasing capability of frontier VLMs~\cite{openai2023gpt4,%
		geminiteam2023,anthropic2024claude} on unconstrained tasks. These
	findings are relevant to estimating per-channel $\tau_h$ values for
	specific challenge families, but they do not bear on whether \BPC's
	cost guarantee holds: that guarantee rests on channel throughput
	alone, and our evaluation (Section~\ref{sec:eval}) shows that
	high per-challenge accuracy does not translate to unbounded throughput
	under deployment-realistic deadlines for any channel type we evaluated.
	
	\noindent\textbf{Classical primitives, individually.}
	It is natural to ask whether \BPC\ is simply a repackaging of
	existing primitives. Taken individually, none of the following
	satisfy~(H1)--(H4) jointly.
	
	\emph{Digital signatures.} A signature
	$\sigma = \mathsf{Sign}(sk_v, m)$ satisfies~(H1)---unforgeability
	binds $\sigma$ to $v$---but says nothing about \emph{when} $\sigma$
	was produced: $v$ can sign today and reveal $\sigma$ years later.
	Composing a signature with a fresh external nonce recovers~(H2) but
	not~(H3): as discussed in Section~\ref{sec:spec}, signing is an
	offline computation with no enforced deadline, so throughput is
	bounded only by computational speed, not by any channel.
	
	\emph{Timestamping.} Digital
	timestamping~\cite{haber1991timestamp} certifies that a document
	existed \emph{before} a given time, typically via a hash chain or
	trusted timestamping authority. This gives an upper bound on when a
	proof was created, not a lower bound: it does not certify that the
	identity was online and responsive at any particular moment, only
	that some artifact predates a point in time. A timestamp can be
	requested for a precomputed artifact at any later moment, violating
	the requirement that a proof's temporal context not be selectable
	by the identity in advance~(H2).
	
	\emph{Verifiable Random Functions.} A
	VRF~\cite{micali1999vrf} produces a pseudorandom value together with
	a publicly verifiable proof of correctness from a secret key and a
	seed---structurally the closest single-shot analogue to
	$\textsf{Prove}$/$\textsf{Verify}$ (Section~\ref{sec:spec}). But a
	VRF evaluation is, like a signature, an offline computation: nothing
	prevents batch-evaluating a VRF over many candidate seeds at once,
	and a VRF alone carries no channel-level throughput bound across many
	keys. It is a plausible component to build $\textsf{Prove}$ from---this
	paper does not---not a substitute for~(H4).
	
	\emph{Remote attestation.} TEE-based
	attestation~\cite{costan2016sgx} certifies that specific code ran on
	specific hardware, which can supply~(H1) (device-bound keys) and,
	with periodic re-attestation, something like~(H2)--(H3). It is a
	legitimate \emph{construction} for Property~(H4)---one of the
	channel-serialization bases discussed in
	Section~\ref{sec:spec}---but it is not itself a statement about
	identity-time cost across many identities: nothing in the
	TEE-attestation literature formalizes a multi-identity throughput
	ceiling or derives a cost theorem from it. The contribution of
	\BPC\ is not any one of these primitives but the composition that
	makes all four properties hold jointly, together with the explicit
	reduction from that composition to a cost theorem
	(Sections~\ref{sec:cost}--\ref{sec:multiwindow}).
	
	\noindent\textbf{Synthesis: what \BPC\ uniquely provides.}
	Table~\ref{tab:hco-taxonomy} organizes prior mechanisms along four
	dimensions critical for sustained Sybil resistance: \emph{persistence}
	(is verification repeated over time?), \emph{per-identity
		rate-limiting} (does each identity face an independent bound?),
	\emph{solver-agnosticism} (does security depend on human superiority
	over AI?), and \emph{joint cost scaling} (does the total adversarial
	cost grow with both $s$ and $T$?). No prior mechanism achieves all
	four simultaneously. Graph-based defenses are not persistent in the
	cost sense. CAPTCHAs and PoP provide one-time $C(s,1)$ bounds but
	not $C(s,T)$ over time. Behavioral biometrics are not per-identity
	rate-limited in the cryptographic sense and are not designed for open
	adversarial settings. Continuous storage proofs formalize persistence
	for a single prover but have no notion of per-identity binding at all.
	Resource-based mechanisms enforce persistence but not per-identity
	independence, and whether the resulting joint cost is sublinear or
	linear is conditional on the economics of sustaining the resource
	rather than fixed by the mechanism's rules alone. \BPC\
	is the first mechanism to unify all four properties in a formally
	specified, solver-agnostic primitive---and the only one that exports
	its verdict as a public, per-window proof $\pi(v,d)$ rather than an
	opaque internal score, with a provable $C(s,T) \geq sT/\tau_h$ lower
	bound.
	
	\begin{table*}[t]
		\centering
		\caption{Comparison of identity verification and anti-Sybil
			mechanisms along four dimensions for sustained Sybil resistance.
			\emph{Persistent}: is verification repeated after admission?
			\emph{Per-identity}: does each identity face an independent bound,
			not shared across a pool?
			\emph{Solver-agnostic}: does security hold without assuming human
			cognitive superiority?
			\emph{Joint cost}: does adversarial cost scale as $\Omega(sT)$
			jointly in identities $s$ and time $T$? Entries marked
			\emph{Conditional} depend on a separate, explicit economic
			assumption about the cost of sustaining the resource, not on
			reusability or parallelizability alone~\cite{maleki2026geometry}.}
		\label{tab:hco-taxonomy}
		\resizebox{\linewidth}{!}{%
			\begin{tabular}{llcccc}
				\toprule
				Mechanism & Representative Systems
				& Persistent & Per-identity & Solver-Agnostic & Joint Cost \\
				\midrule
				Graph-based defenses
				& \cite{yu2006sybilguard,yu2008sybillimit,cao2012syrank}
				& No  & No  & Yes & $O(s)$ only \\
				One-time CAPTCHA
				& \cite{vonahn2003captcha,recaptcha2019,hcaptcha2020}
				& No  & No  & No  & $O(s)$ once \\
				Proof-of-Personhood
				& \cite{borge2017proofofpersonhood,worldcoin2023,semaphore2023}
				& No  & No  & Partial & $O(s)$ once \\
				Behavioral biometrics
				& \cite{frank2013touchalytics,serwadda2013kids,freeman2016onion}
				& Yes & No  & Partial & Unclear \\
				Trusted hardware (TEE)
				& \cite{costan2016sgx,davidson2023privacypass}
				& Yes & Yes & Yes & $\Omega(sT)$ (hardware-bound) \\
				Continuous storage proofs
				& \cite{juels2007pors,shacham2008compact,ateniese2020post}
				& Yes & N/A & Yes & N/A (single-prover) \\
				Resource-based (PoW/PoS)
				& \cite{nakamoto2008bitcoin,kiayias2017ouroboros}
				& Yes & No  & Yes & Conditional \\
				Rate-limited tokens
				& \cite{dwork1993pricing}
				& Partial & No & Yes & $o(sT)$ \\
				\midrule
				\textbf{\BPC\ (this work)}
				& ---
				& \textbf{Yes} & \textbf{Yes} & \textbf{Yes} & $\boldsymbol{\Omega(sT)}$ \\
				\bottomrule
		\end{tabular}}
	\end{table*}
	
	\section{\BPC\ Specification}
	\label{sec:spec}

	\begin{figure}[t]
		\centering
		\resizebox{\linewidth}{!}{%
			\begin{tikzpicture}[
				node distance=6mm,
				box/.style={rectangle, rounded corners, draw, thick, minimum width=2cm,
					minimum height=1.3cm, align=center, font=\bfseries\small, inner sep=4pt},
				lbl/.style={font=\small\itshape, align=center, text width=2.1cm},
				>=Stealth
				]
				\node[box, fill=blue!8, draw=blue!60]  (id)   {Identity\\$v$};
				\node[box, fill=teal!8, draw=teal!60, right=of id]    (win)  {Window $d$\\Index $j$};
				\node[box, fill=purple!8, draw=purple!60, right=of win]  (chal) {\BPC\\challenge\\$\chi_{v,d,j}$};
				\node[box, fill=orange!8, draw=orange!60, right=of chal] (resp) {Response\\$\rho$};
				\node[box, fill=green!8, draw=green!60!black, right=of resp] (ver)  {Verify\\Alg.~1};
				\node[right=8mm of ver, font=\bfseries\large] (out) {\textcolor{green!50!black}{\checkmark}/\textcolor{red}{\ding{55}}};
				
				\draw[->, thick] (id) -- (win);
				\draw[->, thick] (win) -- (chal);
				\draw[->, thick] (chal) -- (resp);
				\draw[->, thick] (resp) -- (ver);
				\draw[->, thick] (ver) -- (out);
				
				\node[lbl, below=2mm of id]   {(H1)\\binding};
				\node[lbl, below=2mm of chal] {(H2)\\freshness};
				\node[lbl, below=2mm of resp] {(H3)\\deadline};
				\node[lbl, below=2mm of ver]  {(H4)\\$\tau_h$ bound};
			\end{tikzpicture}%
		}
		\caption{\BPC\ participation flow. Each identity $v$ requests a
			fresh challenge per window $d$; the verifier binds it to $(v,d,j)$
			via~(H1)--(H2). The participant responds within $\Delta_{\mathrm{resp}}$
			(enforcing~H3); Algorithm~1 verifies correctness and timeliness.
			Property~(H4) bounds how many challenges any single channel
			can service per window.}
		\label{fig:architecture}
	\end{figure}
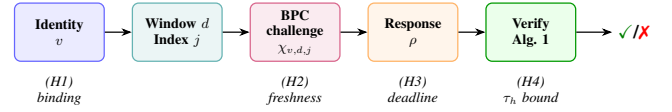
	
	Figure~\ref{fig:architecture} summarizes the resulting participation
	flow, from challenge issuance through verification.
	
	\noindent\textbf{Scope of cryptographic enforcement.}
	Properties~(H1)--(H3) are enforced by the hash-based construction
	of Section~\ref{sec:construction} under standard cryptographic
	assumptions. Property~(H4) is a structural constraint on participation
	channels---enforced either by cognitive serialization (for human
	channels) or by engineering mechanisms such as session serialization
	and device attestation (for automated channels). The security proofs
	of Sections~\ref{sec:cost}--\ref{sec:multiwindow} depend on~(H4)
	as an assumption whose empirical basis for representative challenge
	families is validated in Section~\ref{sec:eval}.
	
	\subsection{Primitive Interface}
	
	Before stating the formal definition, it is useful to separate three
	things that are easy to conflate: the \emph{event} the primitive is
	about, the \emph{claim} a verifier can check, and the \emph{proof}
	that makes the claim checkable without trusting whoever issued the
	qualifying action or received the identity's proof of having
	performed it.
	
	\noindent\textbf{Event.} Identity $v$ performs an action, during
	temporal context $d$ (a \emph{window}), that (i)~requires possession
	of $v$'s credential to perform, and (ii)~could not have been
	performed before $d$ began. What the action concretely consists
	of---a challenge response, a heartbeat, a hardware
	attestation---is left to the realization (Section~\ref{sec:construction});
	only these two structural requirements are part of the primitive.
	
	\noindent\textbf{Claim.} A verifier who accepts $\pi_{v,d}$ may
	conclude: identity $v$---and no other identity---performed the
	qualifying action for context $d$, and did so no earlier than $d$
	permits. This is a claim about what the verifier may believe, not
	merely a description of what occurred.
	
	\noindent\textbf{Proof.} A publicly verifiable object $\pi_{v,d}$
	substantiating the claim. Verification requires no long-term storage
	of participant data and no trust in the issuing party---only the
	identity's registered public key and the public transcript of
	context $d$.
	
	\begin{definition}[Bounded Participation Channel]
		\label{def:hco}
		The \emph{Bounded Participation Channel} (\BPC) is a formal primitive that
		realizes the Event/Claim/Proof pattern above through two operations:
		\[
		\textsf{Prove}: (v,\,d) \;\longrightarrow\; \pi_{v,d},
		\qquad
		\textsf{Verify}: (v,\,d,\,\pi_{v,d}) \;\longrightarrow\; \{0,1\}.
		\]
		At this level, \textsf{Prove} and \textsf{Verify} are left
		abstract: how a proof is produced and what form it takes are
		construction decisions, not part of the primitive's contract. A
		valid proof sets the participation bit $\pi(v,d)=1$
		(Section~\ref{sec:model}); \BPC\ neither stores nor aggregates
		this bit across windows---each window's claim must be
		independently re-established. \BPC\ is solver-agnostic by
		construction: the definition and security analysis impose no
		requirement on what performs the qualifying action---a person, a
		script, or a model---and apply uniformly to any participation
		channel satisfying~(H1)--(H4).
	\end{definition}
	
	Section~\ref{sec:construction} instantiates \textsf{Prove} and
	\textsf{Verify} concretely, as an interactive challenge-response
	exchange; a non-interactive realization (e.g., periodic device
	attestation) would satisfy the same contract without ever issuing a
	challenge. Nothing in this section assumes interactivity.
	
	\subsection{Security Properties}
	
	\BPC's contract is that a proof $\pi_{v,d}$ satisfying \textsf{Verify}
	substantiates the claim above only if the following hold. Properties
	(H1)--(H2) are stated as adversarial objectives---each is directly
	the basis of a security game (Section~\ref{sec:model})---rather than
	as implementation detail; (H3) is a primitive-level qualitative
	requirement, realized concretely in Section~\ref{sec:construction};
	(H4) is a channel-capacity assumption belonging to a different layer
	entirely, stated formally as Assumption~\ref{asm:throughput} in
	Section~\ref{sec:model} and needed only for the multi-identity cost
	theorem of Sections~\ref{sec:cost}--\ref{sec:multiwindow}, not for
	any single $\pi_{v,d}$ to be meaningful.
	
	\begin{itemize}
		\item \textbf{(H1) Identity binding.} No adversary, even given
		valid proofs for identities other than $v$, can produce a $\pi$
		that \textsf{Verify} accepts for $v$ without controlling $v$'s
		secret key material.
		
		\item \textbf{(H2) Non-precomputability.} No adversary can produce
		a $\pi$ that \textsf{Verify} accepts for context $d$ before $d$'s
		context is determined, \emph{and} $d$'s context is not selectable
		or predictable in advance by $v$ itself. The first clause alone is
		not sufficient: an identity that could choose its own context
		trivially satisfies "not valid before a context it picked," while
		gaining nothing in security.
		
		\item \textbf{(H3) Bounded temporal context.} Every valid
		$\pi_{v,d}$ must correspond to an action that occurred within a
		bounded interval---the claim is scoped to $d$, not to "at some
		point." Together with~(H2), boundedness rules out any strategy
		that separates the qualifying action from proof production
		across time. The concrete deadline enforcing this bound is a
		deployment parameter, fixed in Section~\ref{sec:construction},
		not part of the primitive's definition.
		
		\item \textbf{(H4) Per-channel throughput bound.} A structural
		fact about channels, not a property of any single $\pi_{v,d}$:
		stated formally as Assumption~\ref{asm:throughput}
		(Section~\ref{sec:model}) and not enforced by the cryptographic
		construction of Section~\ref{sec:construction}. It is what makes
		the multi-identity cost theorem of
		Sections~\ref{sec:cost}--\ref{sec:multiwindow} possible. The
		structural distinction between temporally reusable resources and
		resources that must be re-acquired in every participation window
		is developed in~\cite{maleki2026geometry}; \BPC\ instantiates the
		latter case through its bounded, per-window participation-channel
		model.
	\end{itemize}
	
	\noindent\textbf{Why not signature + nonce + timestamp?} A plain
	signature over a fresh nonce already satisfies~(H1) and~(H2), but
	not~(H3) with any enforced deadline: signing is an offline
	computation, so nothing stops an identity from producing it at
	leisure, long after the nonce is published. Without a real-time
	deadline, signing throughput is bounded only by computational
	speed, not by any channel constraint, making~(H4) and the joint
	cost theorem vacuous. Property~(H3) is what converts an offline
	computation into a channel-bound, real-time act;
	Section~\ref{sec:related} compares this construction in full
	against signatures, timestamping, and remote attestation.
	
	\begin{proposition}[Security is solver-agnostic]
		\label{prop:agnostic}
		Properties~(H1)--(H2) hold against a PPT adversary
		$\mathcal{A}_{\mathrm{crypto}}$ under the random oracle model
		and EUF-CMA; (H3) holds unconditionally. Under these and
		Assumption~\ref{asm:throughput}~(H4), sustaining $s$ active
		identities over $T$ windows requires
		\[
		C(s,T) \;\geq\; \frac{sT}{\tau_h}
		\;=\; sT\cdot\frac{\underline{\delta}_{\mathrm{rtt}}}{\Delta_{\mathrm{resp}}}
		\]
		independent channel-windows, regardless of whether channels are
		operated by humans, automated agents, or any combination. The
		cost argument depends only on~(H4) and not on computational
		hardness: the bound holds against any adversary---PPT or
		unbounded---given Assumption~\ref{asm:throughput}. This result is
		structural: $C(s,T)$ here counts independent participation
		channel-windows required to sustain a target participation level,
		not an economic expenditure. The relationship between a structural
		requirement of this kind and the separate economic conditions
		under which it becomes a cost floor in currency is developed
		in~\cite{maleki2026geometry}.
	\end{proposition}
	
	The formal proofs of Lemmas~\ref{lem:identity-bound}
	and~\ref{lem:persistence} and
	Theorems~\ref{thm:linear-cost}--\ref{thm:steady-state} appear
	in Sections~\ref{sec:cost}--\ref{sec:multiwindow}.

	\begin{remark}
		Properties~(H1)--(H4) jointly constitute a strictly stronger primitive
		than per-account rate limiting. Rate limiting bounds request frequency
		but provides no cryptographic identity binding, no freshness guarantee,
		and no composable cost theorem. \BPC\ provides all four, enabling the
		$C(s,T) = \Omega(sT)$ result of
		Sections~\ref{sec:cost}--\ref{sec:multiwindow} and composability with
		higher-level protocols that consume its exported participation
		proofs.
	\end{remark}
	
	\subsection{On Property~(H4) and Solver Identity}
	\label{sec:h4-solver}
	
	Property~(H4) asserts that a single channel is throughput-bounded per
	window under $\Delta_{\mathrm{resp}}$. This is \emph{not} equivalent to
	assuming that humans outperform machines on individual challenges, and
	should not be interpreted as a claim that automated systems cannot solve
	\BPC\ challenges. The security argument depends only on bounded throughput
	per participation channel, not on solver identity. Whether channels are
	operated by humans, AI systems, or hybrid pipelines is orthogonal to the
	cost theorem. A single automated agent may solve any given challenge
	instantly; what~(H4) asserts is that a single channel cannot handle
	$\omega(1)$ distinct challenge indices per window within
	$\Delta_{\mathrm{resp}}$, due to the irreducible latency of each
	challenge-response round trip (network delivery, processing, submission,
	and verification). For human channels, cognitive processing time
	contributes to this bound; for automated channels, inference and network
	latency play the same structurally equivalent role. Empirical
	validation of~(H4) for specific challenge families is given in
	Section~\ref{sec:eval}.
	
	\noindent\textbf{Two independent bases for~(H4).}
	Property~(H4) can be grounded in two structurally distinct mechanisms,
	and either one is sufficient for the security argument.
	
	\emph{Cognitive serialization.}
	Human cognition imposes an irreducible serial bottleneck: a single person
	cannot attend to, process, and respond to multiple independent challenges
	simultaneously within a short deadline. This bound derives from cognitive
	and motor throughput limits, not from any assumption of human superiority
	over automated solvers.
	
	\emph{Channel serialization.}
	Engineering constraints can enforce~(H4) independently of who or what
	operates the channel. Concrete mechanisms include:
	\begin{itemize}
		\item \textbf{Per-session serialization.} The system issues at most
		one active challenge per session token at any moment; the browser
		session or WebSocket connection is the channel boundary. A channel
		serving multiple identities in parallel must hold multiple concurrent
		sessions, each incurring its own round-trip budget.
		\item \textbf{Device-bound attestation.} Challenges are bound to
		a specific hardware token, secure enclave, or TPM that enforces
		a per-device rate limit. Scaling to $s$ identities requires $s$
		independently attested devices.
		\item \textbf{Stateful interaction locality.} Challenges require
		continuous interaction with a dynamically mutating UI element
		throughout $\Delta_{\mathrm{resp}}$, making parallel execution across
		identities on a single device structurally difficult.
	\end{itemize}
	
	Under channel serialization, the bound on $\tau_h$ is a systems
	constraint rather than a cognitive claim: it follows from the
	serialization enforced by the channel architecture. The formal security
	argument of Sections~\ref{sec:cost}--\ref{sec:multiwindow} holds under
	either basis for~(H4); deployments may select the enforcement mechanism
	best suited to their trust model and accessibility requirements.
	
	\noindent\textbf{Concurrent versus sequential service.}
	Serialization bounds \emph{concurrent} service only: a channel cannot
	serve multiple identities at the same instant, but nothing prevents it
	from serving multiple identities \emph{sequentially} within one
	window, one round trip after another, up to the response-time bound.
	This is exactly what $\tau_h = \lfloor\Delta_{\mathrm{resp}}/\underline{\delta}_{\mathrm{rtt}}\rfloor$
	(Assumption~\ref{asm:throughput}) counts: the number of sequential
	round trips a single channel can complete before the deadline. Lemma~\ref{lem:identity-bound}
	and the multi-instance accounting of Section~\ref{sec:security} (e.g.,
	334 solver instances sustaining 1{,}000 identities at $\tau_h=3$) both
	rely on this sequential-reuse reading, not on a channel serving only
	one identity per window.
	
	\subsection{Channel Semantics and Enforcement Models}
	\label{sec:channel-semantics}
	
	The word ``channel'' is deliberately abstract in the definition of
	\BPC\ (Section~\ref{sec:spec}), but a concrete deployment must commit
	to what counts as one channel and how~(H4) is enforced for it.
	Table~\ref{tab:channel-semantics} lists the enforcement models
	introduced above, together with the failure mode each one is
	vulnerable to and where in the paper each is instantiated.
	
	\begin{table*}[t]
		\centering
		\caption{Channel enforcement models. \emph{Instantiated in} points to
			where each model is used as a concrete example or evaluated
			empirically; the formal results of
			Sections~\ref{sec:cost}--\ref{sec:multiwindow} are stated
			abstractly over~(H4) and hold under any row of this table.}
		\label{tab:channel-semantics}
		\resizebox{\linewidth}{!}{%
			\begin{tabular}{llll l}
				\toprule
				Channel model & What counts as one channel & Enforcement mechanism & Main limitation & Instantiated in \\
				\midrule
				Human operator
				& One person
				& Cognitive/motor serialization (inherent)
				& Outsourcing to paid human labor markets
				& Section~\ref{sec:h4-solver}; human ranges in Table~\ref{tab:evaluation-results} \\
				Browser session
				& One active authenticated session/connection
				& Server-side session serialization (one open challenge per session)
				& Bot/session farms provisioning many concurrent sessions
				& Section~\ref{sec:h4-solver} \\
				Device-bound channel
				& One attested device or key
				& Hardware attestation (TPM, TEE, passkey)
				& Device farms; emulator/attestation-spoofing resistance
				& Section~\ref{sec:h4-solver}; Table~\ref{tab:hco-taxonomy} \\
				API solver channel
				& One concurrent inference stream (API key/session/quota)
				& API-level rate and concurrency limits
				& Parallel API keys; bulk/tail-latency exploitation
				& Section~\ref{sec:eval}; 334-instance example (Section~\ref{sec:security}) \\
				UI-occupancy channel
				& One continuously occupied interactive surface
				& Sustained interaction required throughout $\Delta_{\mathrm{resp}}$
				& Automation across multiplexed virtual displays
				& Section~\ref{sec:h4-solver} \\
				\bottomrule
			\end{tabular}%
		}
	\end{table*}
	
	Each row fixes a different value of $\tau_h$ but leaves the cost
	theorem itself unchanged: Theorem~\ref{thm:joint-linear} and
	Corollary~\ref{cor:no-sublinear} are proved for an arbitrary channel
	satisfying~(H4) and do not depend on which enforcement model
	supplies it. The empirical evaluation of Section~\ref{sec:eval}
	instantiates the \emph{API solver channel} row specifically---the
	measured $\tau_h$ values in Table~\ref{tab:evaluation-results} are
	not automatically the right calibration for a device-bound or
	UI-occupancy deployment, which may instead fall under the
	exclusive-channel regime ($\tau_h=1$, Remark~\ref{rem:tau-regimes})
	enforced by hardware or interaction constraints rather than by
	response latency.
	
	\subsection{Separability and Composability}
	
	The security guarantees of \BPC\ depend only on~(H1)--(H4). Application-level
	details---user interfaces, challenge formats, delivery channels---do not affect
	these guarantees. This separability enables \BPC\ to be instantiated across
	diverse systems and composed with higher-level protocols, which remain
	free to define their own policies over the exported participation
	proofs $\pi(v,\cdot)$.
	
	\subsection{Delegation and Key Sharing}
	\label{sec:delegation}
	
	\BPC's properties are proofs about a cryptographic key, not about the
	person or process holding it. Property~(H1) ensures that a valid
	response is bound to identity $v$'s registered key $\mathsf{pk}_v$ and
	cannot be credited to any other identity; it does not, and cannot,
	distinguish between $v$ personally producing that response and $v$
	having delegated challenge-solving or shared signing access to a
	third party---a family member, a paid worker, or an automated agent
	acting under $\mathsf{sk}_v$. \BPC's guarantee is therefore precisely
	that sustained access to identity $v$'s signing capability was
	exercised in every counted window, not that a specific human being
	was sustainedly present. This is a deliberate scope choice, consistent
	with \BPC's solver-agnostic design
	(Section~\ref{sec:h4-solver}): the primitive prices sustained
	participation at the channel-window level regardless of what or who
	operates the channel, and does not itself prevent a participant from
	delegating that operation. Deployments that require a stronger
	non-delegation guarantee---binding participation to a specific
	individual, not merely to possession of a specific key---should
	compose \BPC\ with an external mechanism that restricts key custody,
	such as non-exportable hardware-bound keys (secure enclaves, TPMs,
	passkeys) or server-controlled session flows that never expose
	$\mathsf{sk}_v$ to the participant directly (the device-bound channel
	model of Table~\ref{tab:channel-semantics}).

	\section{A Hash-Based Realization}
	\label{sec:construction}
	
	Section~\ref{sec:spec} left $\textsf{Prove}$ and $\textsf{Verify}$
	abstract, specifying only the Event/Claim/Proof contract and
	Properties~(H1)--(H4) that any realization must satisfy. This
	section gives one such realization: a concrete hash-based
	instantiation satisfying Properties~(H1)--(H3) under standard
	assumptions---the random oracle model~\cite{bellare1993random} for
	the hash function, and existential unforgeability under
	chosen-message attack for the signature scheme---via an interactive
	challenge-response exchange. It is not the only realization the
	contract admits: Section~\ref{sec:spec} already noted that periodic
	device attestation would satisfy the same contract without ever
	issuing a challenge, and Section~\ref{sec:related} discusses why
	classical primitives alone---signatures, timestamps, VRFs,
	attestation---do not, individually, suffice. Property~(H4) is a
	channel-level structural constraint that no cryptographic
	construction enforces; it is analyzed separately in
	Section~\ref{sec:model} and empirically validated in
	Section~\ref{sec:eval}. The construction presented here is
	self-contained and requires no trusted hardware.
	
	\begin{figure}[t]
		\centering
		\includegraphics[width=\linewidth]{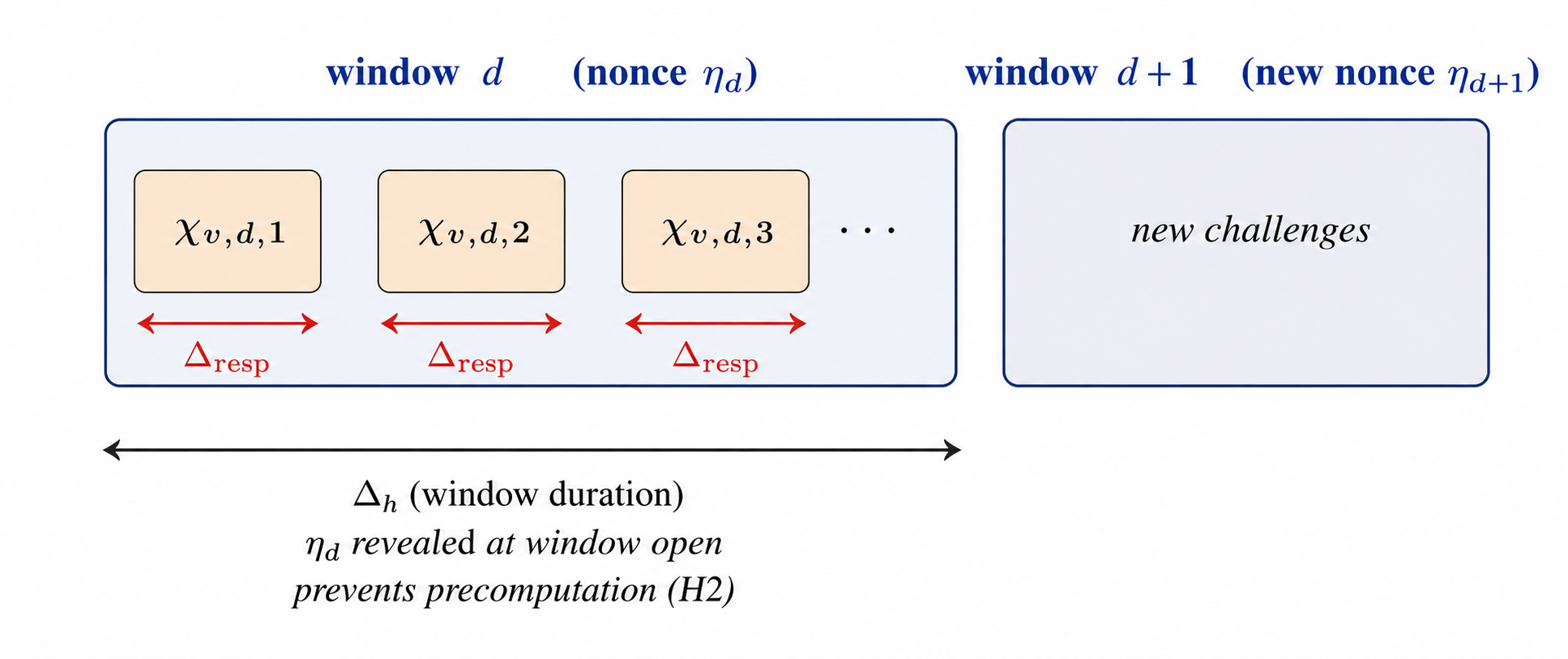}
		\caption{Two-timescale structure. A window of duration $\Delta_h$
			contains multiple challenges $\chi_{v,d,j}$, each with a short
			deadline $\Delta_{\mathrm{resp}} \ll \Delta_h$. A fresh nonce
			$\eta_d$ is drawn at window open, preventing precomputation~(H2);
			the deadline enforces real-time response~(H3). Cross-window
			reuse is ruled out by~(H2): solutions computed in window $d$
			are invalid in $d\!+\!1$.}
		\label{fig:timeline}
	\end{figure}
	
	\textsf{Prove} is realized here as an interactive challenge-response
	exchange, illustrated in Figure~\ref{fig:timeline}: on input a
	participant identity $v$, window index $d \in \mathbb{N}$, and
	challenge index $j \in \{1,\ldots,k(d)\}$, the primitive issues a
	fresh challenge instance $\chi_{v,d,j}$ with response deadline
	$\Delta_{\mathrm{resp}} \ll \Delta_h$. A response $\rho$ to
	$\chi_{v,d,j}$ is \emph{valid} iff: (i)~$\rho$ is received within
	$\Delta_{\mathrm{resp}}$ of issuance (realizing~H3); (ii)~$\rho$ is a
	correct solution to $\chi_{v,d,j}$; (iii)~$\rho$ is cryptographically
	bound to the triple $(v,d,j)$ (realizing~H1). A valid response
	constitutes the proof $\pi_{v,d}$.
	
	\subsection{Setup}
	
	Let $\mathcal{H}: \{0,1\}^* \to \{0,1\}^\lambda$ be a collision-resistant
	hash function modeled as a random oracle~\cite{bellare1993random}, with
	security parameter $\lambda$. We instantiate $\mathcal{H}$ with
	SHA-3~\cite{nist2015sha3} in practice. Each identity $v$ holds a
	long-term key pair $(\mathsf{sk}_v, \mathsf{pk}_v)$ with $\mathsf{pk}_v$
	registered with the system. The registration mechanism is external
	to \BPC\ and may be any identity management system---Proof-of-Personhood
	protocols~\cite{borge2017proofofpersonhood}, institutional enrollment,
	or payment-gated admission. \BPC\ does not require registration to be
	costly: its linear-cost guarantee governs the \emph{sustained
		participation} cost regardless of how cheap identity creation is. An
	adversary registering $s$ identities incurs the deploying system's
	per-identity admission cost, which is orthogonal to and additive with
	the per-window participation cost $C(s,T)$ enforced by \BPC.
	
	At the start of each window $d$, a fresh nonce
	$\eta_d \in \{0,1\}^\lambda$ is drawn from a public randomness
	beacon~\cite{syta2017drand} and publicly broadcast. The nonce is
	unpredictable before window $d$ opens and publicly verifiable thereafter.
	
	\subsection{Challenge Generation}
	
	The binding tag alone does not make the \emph{content} of a challenge
	fresh: if the task payload $\phi_{v,d,j}$ were drawn from a source
	independent of $\eta_d$---a pre-generated pool, say---an adversary
	could solve payloads before window $d$ opens even though the tag
	itself remains unpredictable, defeating the intent of~(H2). To rule
	this out, both the tag and the payload are derived from one shared,
	window-fresh seed. For identity $v$, window $d$, and index
	$j \in \{1,\ldots,k(d)\}$:
	\begin{align}
		\mathrm{seed}_{v,d,j} &\;=\;
		\mathcal{H}(\mathsf{pk}_v \;\|\; d \;\|\; j \;\|\; \eta_d),
		\label{eq:seed-gen}
		\\[2pt]
		(\phi_{v,d,j},\, \mathsf{aux}_{v,d,j}) &\;\longleftarrow\;
		\mathsf{Gen}_{\mathcal{C}}(\mathrm{seed}_{v,d,j}),
		\label{eq:challenge-gen}
	\end{align}
	where $\mathsf{Gen}_{\mathcal{C}}$ is the (possibly randomized-looking
	but seed-deterministic) instance generator of challenge family
	$\mathcal{C}$ from Section~\ref{sec:challenges}. $\mathrm{seed}_{v,d,j}$
	is exactly the binding tag recomputed in
	Algorithm~\ref{alg:verify}---no new value is introduced. Because
	$\phi_{v,d,j}$ is a deterministic function of that same unpredictable
	$\eta_d$, freshness of the tag under~(H2) carries over to freshness
	of the challenge content---an adversary cannot precompute
	$\phi_{v,d,j}$ any earlier than it can precompute the tag.
	$\mathsf{aux}_{v,d,j}$ denotes any auxiliary data the generator
	produces alongside the payload (e.g., which candidate is correct);
	whether $\mathsf{aux}_{v,d,j}$ is exposed to the participant is a
	property of the family, discussed next.
	
	\noindent\textbf{Verification is public, by design.} $\mathsf{CheckTask}$
	in Algorithm~\ref{alg:verify} takes only the public payload
	$\phi_{v,d,j}$ and the submitted response $r_{v,d,j}$---never a
	secret answer key held only by the verifier. This is a deliberate
	consequence of \BPC's solver-agnostic scope
	(Section~\ref{sec:h4-solver}): the security argument never rests on
	a challenge being hard to solve, only on a response being
	identity-bound, fresh, and rate-limited, so nothing is lost by making
	correctness fully publicly checkable. This does mean that admissible
	challenge families (Section~\ref{sec:challenges}) must be ones whose
	correctness predicate is public-computable from $(\phi_{v,d,j},
	r_{v,d,j})$ alone---a family that instead required a
	verifier-only secret to check responses would fall outside the
	scope of this construction and is not considered here.
	
	\subsection{Response and Verification}
	
	A participant submits:
	\[
	\rho \;=\; \bigl(\sigma_v,\; r_{v,d,j}\bigr),
	\]
	where $r_{v,d,j} \in \mathcal{R}$ is the task solution and
	$\sigma_v = \mathsf{Sign}(\mathsf{sk}_v,\;
	\mathsf{pk}_v \;\|\; d \;\|\; j \;\|\; r_{v,d,j})$.
	Algorithm~\ref{alg:verify} gives the complete verification procedure.
	
	\begin{algorithm}[t]
		\caption{$\BPC.\mathsf{Verify}(v,\, d,\, j,\, \rho,\, t_{\mathrm{recv}})$}
		\label{alg:verify}
		\begin{algorithmic}[1]
			\REQUIRE Identity $v$, window $d$, index $j$,
			response $\rho = (\sigma_v, r_{v,d,j})$,
			receipt time $t_{\mathrm{recv}}$
			\ENSURE \textbf{Accept} or \textbf{Reject}
			\IF{$t_{\mathrm{recv}} > t_{\mathrm{issue}} + \Delta_{\mathrm{resp}}$}
			\RETURN \textbf{Reject} \hfill\COMMENT{Deadline exceeded --- enforces (H3)}
			\ENDIF
			\IF{$\lnot\;\mathsf{VerifySig}\!\bigl(\mathsf{pk}_v,\;
				\mathsf{pk}_v \;\|\; d \;\|\; j \;\|\; r_{v,d,j},\; \sigma_v\bigr)$}
			\RETURN \textbf{Reject} \hfill\COMMENT{Invalid signature --- enforces (H1)}
			\ENDIF
			\STATE Recompute binding tag:
			$\tau \leftarrow \mathcal{H}(\mathsf{pk}_v \;\|\; d \;\|\; j \;\|\; \eta_d)$
			\IF{binding tag in $\chi_{v,d,j}$ $\neq$ $\tau$}
			\RETURN \textbf{Reject} \hfill\COMMENT{Wrong identity/window --- enforces (H1)}
			\ENDIF
			\IF{$\lnot\;\mathsf{CheckTask}(r_{v,d,j},\; \phi_{v,d,j})$}
			\RETURN \textbf{Reject} \hfill\COMMENT{Incorrect task solution}
			\ENDIF
			\RETURN \textbf{Accept}
		\end{algorithmic}
	\end{algorithm}
	
	\subsection{Security of the Construction}
	
	We establish each property against the adversary model of
	Section~\ref{sec:model}. Properties~(H1)--(H2) hold against the
	PPT adversary $\mathcal{A}_{\mathrm{crypto}}$; (H3) holds
	unconditionally; (H4) follows from
	Assumption~\ref{asm:throughput} against any adversary.
	
	\noindent\textbf{Identity binding~(H1).}
	Against a PPT adversary $\mathcal{A}_{\mathrm{crypto}}$, the
	binding tag $\mathcal{H}(\mathsf{pk}_v \;\|\; d \;\|\; j \;\|\;
	\eta_d)$ ties the challenge to $(v,d,j)$. A valid response
	includes a signature under $\mathsf{sk}_v$; producing a valid
	response for $v' \neq v$ requires either forging a signature under
	$\mathsf{sk}_v$---infeasible under EUF-CMA---or finding a hash
	collision---infeasible in the random oracle
	model~\cite{bellare1993random}.
	
	\noindent\textbf{Freshness~(H2).}
	Against a PPT adversary $\mathcal{A}_{\mathrm{crypto}}$, the nonce
	$\eta_d$ is revealed only at the start of window $d$. Because $\eta_d$
	is hidden prior to window $d$, the binding tag
	$\mathcal{H}(\mathsf{pk}_v \|\, d \|\, j \|\, \eta_d)$ is
	computationally indistinguishable from uniform before $d$ opens.
	No PPT adversary can compute $\chi_{v,d,j}$ before $\eta_d$ is
	known, ruling out precomputation in the random oracle
	model~\cite{bellare1993random}.
	
	\noindent\textbf{Real-time constraint~(H3).}
	This property holds \emph{unconditionally}, independently of
	adversarial computational power. Line~1 of
	Algorithm~\ref{alg:verify} checks the receipt timestamp
	deterministically and independently of the task solution.
	Late responses are rejected regardless of the adversary's
	computational capabilities.
	
	\noindent\textbf{Throughput bound~(H4).}
	The cryptographic construction does not enforce~(H4); this property
	is captured by Assumption~\ref{asm:throughput} and arises from the
	latency structure of the participation channel. Each valid response
	requires: (i)~receiving $\chi_{v,d,j}$ after its issuance;
	(ii)~computing $r_{v,d,j}$ and $\sigma_v$; (iii)~transmitting $\rho$
	and having it received within $\Delta_{\mathrm{resp}}$. A
	deployment-specific lower bound $\underline{\delta}_{\mathrm{rtt}}$ on
	this round-trip time (Assumption~\ref{asm:throughput}) limits any
	channel to at most
	$\tau_h = \lfloor\Delta_{\mathrm{resp}}/\underline{\delta}_{\mathrm{rtt}}\rfloor$
	responses per window. For human channels, cognitive processing time
	contributes to this floor. The security results of
	Sections~\ref{sec:cost}--\ref{sec:multiwindow} are stated conditionally
	on Assumption~\ref{asm:throughput}; empirical characterization of
	$\tau_h$ for representative challenge families is given in
	Section~\ref{sec:eval}.

	\section{System and Adversary Model}
	\label{sec:model}
	
	\subsection{System Model}
	
	The system consists of a set of identities $\mathcal{I}$ interacting with
	a platform relying on \BPC\ for continuous verification. Time is divided
	into windows $d \in \mathbb{N}$ of fixed duration $\Delta_h$.
	
	For identity $v$ and window $d$, let $\pi(v,d) \in \{0,1\}$ denote the
	\emph{participation bit}: $\pi(v,d)=1$ iff $v$ performs at least one
	valid qualifying action during $d$, and $\pi(v,d)=0$ otherwise. An
	identity $v$ is \emph{active} in window $d$ iff $\pi(v,d)=1$.
	\BPC's internal state supports only the cryptographic bookkeeping
	required for freshness---nonces and window indices (H2)---and does
	not accumulate $\pi(v,\cdot)$ into a score, reputation, or any other
	persistent participation record. \BPC\ exports the sequence
	$\pi(v,1),\pi(v,2),\ldots$ as its sole output; interpreting or
	aggregating this sequence for reputation, voting weight, or
	consensus influence is the responsibility of the higher-level
	system that consumes it, not of \BPC\ itself. Verification is
	deterministic, public, and requires no long-term storage of
	participant data.
	
	\subsection{Adversary Model}
	
	We consider a two-layer adversary model reflecting the two-layer
	structure of \BPC's security guarantees.
	
	\noindent\textbf{Cryptographic layer~((H1)--(H3)).}
	For the purposes of the construction in Section~\ref{sec:construction},
	we consider a probabilistic polynomial-time~(PPT) adversary
	$\mathcal{A}_{\mathrm{crypto}}$ bounded by the security parameter
	$\lambda$. Properties~(H1) and~(H2) hold against
	$\mathcal{A}_{\mathrm{crypto}}$ under the random oracle model and
	existential unforgeability of the signature scheme~(EUF-CMA).
	Property~(H3) holds unconditionally: the deadline check in
	Algorithm~\ref{alg:verify} is deterministic and requires no
	computational assumption.
	
	\noindent\textbf{Structural layer~((H4)).}
	For the cost analysis of
	Sections~\ref{sec:cost}--\ref{sec:multiwindow}, we consider a
	computationally unbounded adversary $\mathcal{A}$ that may:
	\begin{itemize}
		\item create and coordinate an arbitrary number of identities $s$;
		\item operate up to $m$ independent participation channels (human or
		automated) within any window, with $m$ adaptively varying across windows;
		\item employ outsourcing, relay attacks, AI automation, or any combination;
		\item deviate arbitrarily from the protocol, subject to
		Assumption~\ref{asm:throughput}~(H4).
	\end{itemize}
	The cost results of Sections~\ref{sec:cost}--\ref{sec:multiwindow}
	hold unconditionally given Assumption~\ref{asm:throughput},
	independently of adversarial computational power.
	
	\subsection{Adversarial Cost Function}
	
	\begin{definition}[Adversarial cost]
		\label{def:cost}
		$C(s, T)$ denotes the minimum number of independent participation
		channel-windows required to sustain $s$ active identities over $T$
		consecutive windows.
	\end{definition}
	
	\begin{assumption}[Channel Throughput Bound]
		\label{asm:throughput}
		Let $\underline{\delta}_{\mathrm{rtt}} > 0$ denote a
		\emph{deployment-specific lower bound} on the time required by a
		single participation channel to complete one valid
		challenge-response interaction under the deployed channel
		model---a round trip consisting of challenge delivery, response
		computation, and receipt acknowledgement---such that no channel
		admitted by that model can complete a round trip faster than
		$\underline{\delta}_{\mathrm{rtt}}$. Any single participation
		channel then produces at most
		\[
		\tau_h \;=\; \left\lfloor \frac{\Delta_{\mathrm{resp}}}{\underline{\delta}_{\mathrm{rtt}}} \right\rfloor
		\]
		valid responses per window, independently of the number of
		identities associated with that channel. This is a structural
		claim about the channel model, not a statistical one:
		$\underline{\delta}_{\mathrm{rtt}}$ must hold against the
		\emph{fastest} channel the deployment's threat model admits, not
		merely against a typical one, so it cannot simply be read off as
		an average of observed latencies. Section~\ref{sec:eval}
		estimates $\underline{\delta}_{\mathrm{rtt}}$ conservatively from
		observed latency distributions rather than instantiating it with
		the sample mean; the gap between the two, and the risk of a real
		deployment's fastest achievable channel undercutting whatever
		value of $\underline{\delta}_{\mathrm{rtt}}$ is adopted, is
		treated explicitly as a deployment consideration in
		Section~\ref{sec:limitations} rather than as a hidden gap.
	\end{assumption}
	
	\begin{remark}[Three capacity regimes]
		\label{rem:tau-regimes}
		The quantity $\tau_h$ in Assumption~\ref{asm:throughput} names one
		of three structurally distinct capacity regimes; every other
		occurrence of $\tau_h$ in this paper refers to the first.
		\begin{itemize}
			\item \textbf{Synchronous deadline capacity},
			$\tau_h^{\mathrm{sync}} = \lfloor
			\Delta_{\mathrm{resp}}/\underline{\delta}_{\mathrm{rtt}}
			\rfloor$: the number of sequential challenge-response round
			trips a single channel can complete before a shared deadline
			$\Delta_{\mathrm{resp}}$ closes. This is the regime instantiated
			by the browser-session and API solver channel models of
			Table~\ref{tab:channel-semantics}, the one used in
			Theorem~\ref{thm:joint-linear} and
			Corollary~\ref{cor:no-sublinear}, and the one measured
			empirically in Section~\ref{sec:eval}.
			\item \textbf{Window-service capacity}, $\tau_h^{\mathrm{win}} =
			\lfloor \Delta_h/\delta_{\mathrm{service}} \rfloor$, where
			$\Delta_h$ is the full window length rather than a single
			response deadline and $\delta_{\mathrm{service}}$ is the time a
			channel takes to service one identity's request within it: the
			regime for a channel that stays continuously open across the
			whole window rather than answering one burst deadline.
			$\tau_h^{\mathrm{win}}$ can exceed $\tau_h^{\mathrm{sync}}$ when
			$\Delta_h \gg \Delta_{\mathrm{resp}}$; a deployment operating in
			this regime re-derives Corollary~\ref{cor:no-sublinear} with
			$\tau_h^{\mathrm{win}}$ in place of $\tau_h^{\mathrm{sync}}$---the
			algebra is unchanged, only the capacity value differs.
			\item \textbf{Exclusive-channel capacity}, $\tau_h = 1$: the
			degenerate case in which the enforcement mechanism itself, not
			response latency, limits a channel to one identity per window
			regardless of $\Delta_{\mathrm{resp}}$ or
			$\underline{\delta}_{\mathrm{rtt}}$. This is the regime
			instantiated by the device-bound and UI-occupancy channel
			models of Table~\ref{tab:channel-semantics}: a single attested
			device or a single continuously occupied interactive surface
			cannot serve a second identity within the same window at all,
			so the round-trip formula above neither applies nor is needed---%
			$\tau_h = 1$ holds by construction.
		\end{itemize}
		All results in Sections~\ref{sec:cost}--\ref{sec:multiwindow} are
		proved for a generic $\tau_h$ and hold under any of the three
		regimes; only the numeric value, and for the window-service regime
		the substituted formula, changes.
	\end{remark}
	
	\begin{remark}[Bases for Assumption~\ref{asm:throughput}]
		\label{rem:h4-bases}
		Assumption~\ref{asm:throughput} can be grounded in two
		structurally distinct mechanisms.
		\emph{(i)~Cognitive serialization:} for human channels,
		cognitive and motor processing time sets
		$\underline{\delta}_{\mathrm{rtt}} \geq \delta_{\mathrm{cog}}$, where
		$\delta_{\mathrm{cog}}$ is the irreducible human response latency
		for the deployed challenge family (empirically characterized
		in Section~\ref{sec:eval}).
		\emph{(ii)~Channel serialization:} for automated channels,
		$\underline{\delta}_{\mathrm{rtt}}$ is bounded below by network
		propagation delay plus inference latency. Measured under standard
		commercial inference APIs (Section~\ref{sec:eval}), current
		frontier VLMs incur mean latencies of $2.25$--$3.57\,\mathrm{s}$ on
		multi-image perceptual tasks and $0.95$--$1.79\,\mathrm{s}$ on
		text-based reasoning tasks; treating these observed means as an
		illustrative (not conservative) reference point yields $\tau_h \in
		\{2,3\}$ for $\Delta_{\mathrm{resp}} = 8\,\mathrm{s}$ on perceptual
		challenges and $\tau_h \in \{6,7,12\}$ for
		$\Delta_{\mathrm{resp}} = 12\,\mathrm{s}$ on reasoning challenges;
		Section~\ref{sec:eval-results} discusses conservative calibration
		of $\underline{\delta}_{\mathrm{rtt}}$ itself.
		Either basis is sufficient for the security argument of
		Sections~\ref{sec:cost}--\ref{sec:multiwindow}.
	\end{remark}
	
	\section{Sybil Cost Analysis}
	\label{sec:cost}
	
	Figure~\ref{fig:cost-scaling} previews where this section's result
	(and Section~\ref{sec:multiwindow}'s multi-window extension of it)
	sits relative to the mechanisms surveyed in Section~\ref{sec:related}.
	
	\begin{figure}[t]
		\centering
		\begin{tikzpicture}[font=\footnotesize]
			\begin{scope}
				\draw[->] (0,0) -- (4.6,0) node[right] {identities $s$};
				\draw[->] (0,0) -- (0,3.2) node[above] {adversarial cost $C(s,T)$};
				
				\fill[red!25, opacity=0.6]
				(0,0) -- (4.2,2.6) -- plot[domain=4.2:0, samples=60]
				(\x, {1.05*ln(\x+1)}) -- cycle;
				\draw[red!60!black, thick, dashed] (0,0) -- (4.2,2.6);
				\draw[red!60!black, thick, dotted]
				plot[domain=0:4.2, samples=60]
				(\x, {1.05*ln(\x+1)});
				\node[right, font=\scriptsize, red!60!black, align=left]
				at (4.25,1.5)
				{PoW/PoS: $\Omega(s)$ \emph{or} $o(s)$\\
					(depends on carrying cost)};
				
				\draw[blue!70!black, very thick] (0,0) -- (4.2,3.0)
				node[right, font=\scriptsize, blue!70!black]
				{\textbf{BPC}: $\Omega(s)$ per window};
				
				\draw[orange!70!black, thick, dotted] (0,0.9) -- (4.2,0.9)
				node[right, font=\scriptsize, orange!70!black]
				{CAPTCHA/PoP: $O(1)$};
				
				\node[font=\scriptsize\itshape, gray] at (2.1,-0.45)
				{(fixed time horizon $T$)};
			\end{scope}
		\end{tikzpicture}
		\caption{Schematic adversarial cost $C(s,T)$ as a function of the
			number of sustained identities $s$ (at fixed $T$). \BPC\ enforces
			linear scaling $\Omega(s)$ per window unconditionally, from a
			structural channel constraint alone. Resource-based mechanisms
			such as PoW/PoS occupy the shaded band: their regime is
			$\Omega(s)$ or $o(s)$ depending on a separate, explicit economic
			assumption about the cost of sustaining the resource, not on
			reusability or parallelizability alone~\cite{maleki2026geometry};
			one-time mechanisms (CAPTCHA, PoP) incur constant cost after
			creation regardless of $s$.}
		\label{fig:cost-scaling}
	\end{figure}
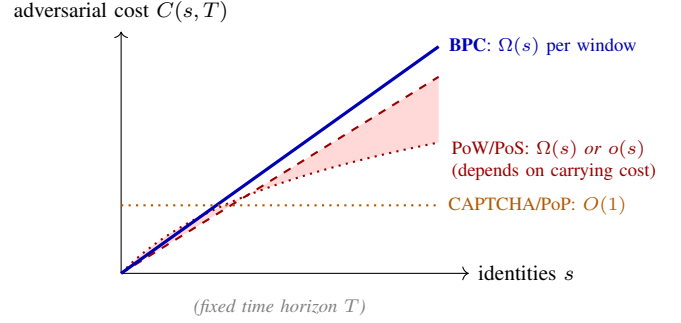
	
	\begin{lemma}[Per-window identity bound]
		\label{lem:identity-bound}
		Under Assumption~\ref{asm:throughput}, in any window $d$, an adversary
		operating $m$ participation channels can produce valid \BPC\ responses
		for at most
		\[
		m \cdot \tau_h \;=\; m \cdot
		\left\lfloor\frac{\Delta_{\mathrm{resp}}}{\underline{\delta}_{\mathrm{rtt}}}\right\rfloor
		\]
		distinct identities.
	\end{lemma}
	
	\begin{proof}
		By Assumption~\ref{asm:throughput}~(H4), each of $m$ channels produces
		at most $\tau_h = \lfloor\Delta_{\mathrm{resp}}/\underline{\delta}_{\mathrm{rtt}}\rfloor$
		valid responses per window. By~(H1), a response for identity $v$ cannot
		be credited to $v' \neq v$; responses are non-sharable across identities.
		By~(H2) and~(H3), precomputed or late responses are rejected.
		The total valid responses attributable to distinct identities is
		therefore at most $m \cdot \tau_h$.
	\end{proof}
	
	\begin{theorem}[Linear Sybil cost]
		\label{thm:linear-cost}
		Under~(H1)--(H3) and Assumption~\ref{asm:throughput}~(H4), sustaining
		$s$ active identities in any single window requires
		\[
		m \;\geq\; \left\lceil\frac{s}{\tau_h}\right\rceil
		\;=\; \left\lceil s \cdot
		\frac{\underline{\delta}_{\mathrm{rtt}}}{\Delta_{\mathrm{resp}}}\right\rceil
		\]
		independent participation channels. Equivalently,
		$C(s,1) \geq s/\tau_h = s \cdot \underline{\delta}_{\mathrm{rtt}}/\Delta_{\mathrm{resp}}$.
	\end{theorem}
	
	\begin{proof}
		From Lemma~\ref{lem:identity-bound}, $m\tau_h \ge s$, giving
		$m \ge s/\tau_h = s \cdot \underline{\delta}_{\mathrm{rtt}}/\Delta_{\mathrm{resp}}$.
		Since $\tau_h = \lfloor\Delta_{\mathrm{resp}}/\underline{\delta}_{\mathrm{rtt}}\rfloor$
		is fixed under Assumption~\ref{asm:throughput}, this gives $m = \Omega(s)$.
		By Definition~\ref{def:cost}, $C(s,1)$ is exactly this minimum
		channel count $m$, so $C(s,1) \geq s/\tau_h$.
	\end{proof}
	
	\begin{remark}[On the simplicity of this reduction]
		The proof above is a direct counting argument, not a computational
		reduction, and this is intentional rather than a gap. Once~(H1)--(H4)
		are granted, linear cost follows by pigeonhole; the same is true of
		many security reductions once the underlying hardness assumption is
		granted---IND-CPA security of ElGamal follows quickly from DDH, for
		instance. The intellectual content of such results lies in identifying
		a minimal, achievable assumption and showing it implies the target
		guarantee, not in the length of the derivation. Here, the analogous
		work is threefold: showing that~(H1)--(H4) are \emph{sufficient}
		for linear cost, while making explicit the role each property
		plays in the bound---(H1) rules out crediting one response to
		multiple identities, (H2)--(H3) rule out precomputed or late
		responses, and (H4) caps per-channel throughput---and situating
		this combination against prior mechanisms that achieve at most a
		subset of it (Section~\ref{sec:related}); constructing~(H1)--(H3)
		from standard cryptographic primitives (Section~\ref{sec:construction});
		and empirically establishing that~(H4) actually holds, with small
		$\tau_h$, for realistic channels---which was not obvious in advance
		and is where this paper's claims could have failed
		(Section~\ref{sec:eval}). We do not claim~(H1)--(H4) are
		individually \emph{necessary} for linear cost---a different
		mechanism might achieve the same scaling via a different property
		combination---only that this combination suffices and that no
		prior mechanism surveyed in Section~\ref{sec:related} attains it
		with fewer properties. The theorem is simple; whether its
		hypothesis is true of the physical world is not, and that is the
		question the empirical evaluation is designed to answer.
	\end{remark}
	
	\begin{corollary}[No sublinear amortization]
		\label{cor:no-sublinear}
		No adversary can sustain $s$ identities with $o(s)$ channels per window,
		regardless of the value of $\tau_h$.
	\end{corollary}
	
	\begin{proof}
		By Lemma~\ref{lem:identity-bound}, each of $m$ channels produces valid
		responses for at most $\tau_h$ distinct identities per window~(H4).
		Property~(H1) ensures that a response for identity $v$ cannot be
		credited to $v' \neq v$; responses are non-transferable. Therefore
		$m$ channels cover at most $m\tau_h$ identities, giving
		$m \geq \lceil s/\tau_h \rceil$. Since $\tau_h$ is a fixed positive
		constant under Assumption~\ref{asm:throughput}, $\lceil s/\tau_h
		\rceil = \Omega(s)$, so no adversary can sustain $s$ identities with
		$o(s)$ channels per window regardless of the value of $\tau_h$.
	\end{proof}
	
	\begin{remark}[Structural vs.\ economic amortization]
		Corollary~\ref{cor:no-sublinear} rules out amortization in a
		\emph{structural} sense: no adversary, regardless of scale, can
		cover more identities per channel as $s$ grows---the
		channel-to-identity ratio never improves with size. It does not by
		itself imply that the \emph{dollar} cost of operating those
		channels resists amortization: bulk pricing or labor-market
		economies of scale could still drive down the cost per channel as
		an adversary scales. That gap is closed only under the separate,
		explicit Assumption~\ref{asm:opcost} and Corollary~\ref{cor:concrete}
		below, which convert this structural count into an economic floor.
	\end{remark}

	\section{Multi-Window and Long-Term Analysis}
	\label{sec:multiwindow}
	
	Let $s_d$ and $m_d$ denote the number of active identities and channels,
	respectively, in window $d$.
	
	\begin{lemma}[Window persistence bound]
		\label{lem:persistence}
		Under Assumption~\ref{asm:throughput},
		$s_d \le m_d \cdot \tau_h = m_d \cdot
		\lfloor\Delta_{\mathrm{resp}}/\underline{\delta}_{\mathrm{rtt}}\rfloor$
		for every window $d$.
	\end{lemma}
	
	\begin{proof}
		Lemma~\ref{lem:identity-bound} applied to window $d$. Since~(H2)--(H3)
		prevent cross-window reuse, effort in window $d$ cannot carry to $d+1$.
	\end{proof}
	
	\begin{theorem}[Joint linear cost]
		\label{thm:joint-linear}
		Under~(H1)--(H3) and Assumption~\ref{asm:throughput}~(H4),
		\[
		C(s,T) \;\geq\; \frac{sT}{\tau_h}
		\;=\; sT \cdot \frac{\underline{\delta}_{\mathrm{rtt}}}{\Delta_{\mathrm{resp}}}.
		\]
	\end{theorem}
	
	\begin{proof}
		By Lemma~\ref{lem:persistence}, each window $d$ independently requires
		at least $\lceil s/\tau_h \rceil$ channel-windows. Properties~(H2)--(H3)
		ensure that valid responses from window $d$ are rejected in window
		$d+1$: each window requires fresh participation and cannot draw on
		effort invested in prior windows. The channel-window counts across
		windows are therefore additive, and the total cost is
		$\sum_{d=1}^{T}\lceil s/\tau_h \rceil \geq sT/\tau_h$.
	\end{proof}
	
	Theorem~\ref{thm:joint-linear} bounds $C(s,T)$, a \emph{structural}
	count of channel-windows (Definition~\ref{def:cost}); it makes no
	claim about currency. Converting this structural requirement into an
	economic floor takes one further, separately falsifiable step.
	
	\begin{assumption}[Uniform Per-Channel Operating Cost]
		\label{asm:opcost}
		There exists a constant $c_{\mathrm{op}} > 0$, independent of $s$
		and $T$, such that provisioning and operating each channel-window
		required by Theorem~\ref{thm:joint-linear} incurs an economic cost
		of at least $c_{\mathrm{op}}$. This is an economic claim, not a
		structural one: bulk pricing or labor-market economies of scale
		available only to a large-$s$ adversary would violate it, an
		empirical question about a given deployment addressed in
		Section~\ref{sec:limitations}, not a consequence of~(H1)--(H4).
	\end{assumption}
	
	\begin{corollary}[Concrete security bound]
		\label{cor:concrete}
		Let $C_{\mathrm{econ}}(s,T)$ denote the minimum economic cost, in
		currency, of sustaining $s$ identities over $T$ windows. Under
		Assumption~\ref{asm:opcost},
		\[
		C_{\mathrm{econ}}(s,T) \;\geq\; c_{\mathrm{op}} \cdot C(s,T) \;=\; \Omega(sT).
		\]
		For illustration (using the cross-solver mean latencies of
		Section~\ref{sec:eval} only as a convenient reference point;
		Section~\ref{sec:eval-results} gives the conservative calibration
		an actual deployment should use instead), take a perceptual
		deployment with $\Delta_{\mathrm{resp}} = 8\,\mathrm{s}$ and
		$\underline{\delta}_{\mathrm{rtt}} \approx 2.7\,\mathrm{s}$:
		$\tau_h = \lfloor 8/2.7 \rfloor = 2$, giving
		$C(s,T) \geq sT/2$ channel-windows and
		$C_{\mathrm{econ}}(s,T) \geq c_{\mathrm{op}} \cdot sT/2$.
		For a reasoning deployment with $\Delta_{\mathrm{resp}} = 12\,\mathrm{s}$
		and $\underline{\delta}_{\mathrm{rtt}} \approx 1.4\,\mathrm{s}$,
		giving $\tau_h = 8$: the bound relaxes to
		$C_{\mathrm{econ}}(s,T) \geq c_{\mathrm{op}} \cdot sT/8$. In both
		regimes $\tau_h$ is a small single-digit constant and the $\Omega(sT)$
		scaling of $C(s,T)$ is preserved unconditionally; the scaling of
		$C_{\mathrm{econ}}(s,T)$ additionally requires
		Assumption~\ref{asm:opcost}. Reducing $\Delta_{\mathrm{resp}}$
		decreases $\tau_h$ proportionally, tightening the per-channel-window
		guarantee at the cost of higher interaction burden.
	\end{corollary}
	
	\begin{theorem}[Steady-state Sybil capacity]
		\label{thm:steady-state}
		Under Assumption~\ref{asm:throughput}, let
		$\bar{m} = \limsup_{T\to\infty}\frac{1}{T}\sum_{d=1}^{T}m_d$.
		Then
		\[
		\limsup_{T\to\infty}\frac{1}{T}\sum_{d=1}^{T}s_d
		\;\le\; \bar{m}\cdot\tau_h
		\;=\; \bar{m}\cdot
		\left\lfloor\frac{\Delta_{\mathrm{resp}}}{\underline{\delta}_{\mathrm{rtt}}}\right\rfloor.
		\]
	\end{theorem}
	
	\begin{proof}
		Sum $s_d \le m_d\tau_h$ over $T$ windows, divide by $T$, take
		$\limsup$. The bound is explicit in $\Delta_{\mathrm{resp}}$ and
		$\underline{\delta}_{\mathrm{rtt}}$ via Assumption~\ref{asm:throughput}.
	\end{proof}
	
	\noindent\textbf{Burst attacks.}
	Concentrating effort in short bursts does not help: the total
	identity-window budget $B_T = \sum_{d=1}^T s_d$ satisfies
	$B_T \le \tau_h\sum_{d=1}^T m_d$, so burst attacks incur linear
	cost in total channel-windows regardless of temporal concentration.

	\section{Challenge Families}
	\label{sec:challenges}
	
	\subsection{Abstract Admissibility Requirements}
	
	A challenge family $\mathcal{C}$ is \emph{admissible} for \BPC\ if
	every instance satisfies~(H1)--(H4). Concretely:
	\begin{itemize}
		\item instances are unpredictable before issuance~(H2), which
		requires both the binding tag \emph{and} the task payload
		$\phi_{v,d,j}$ to be generated from the shared window-fresh
		$\mathrm{seed}_{v,d,j}$ of Equation~\ref{eq:seed-gen}, not from an
		independent or pre-generated source;
		\item correct responses are deterministically and
		\emph{publicly} verifiable from $(\phi_{v,d,j}, r_{v,d,j})$ alone,
		with no verifier-only secret required by $\mathsf{CheckTask}$;
		\item responses are cryptographically bound to $(v,d,j)$~(H1)
		via the construction of Section~\ref{sec:construction};
		\item the interaction model imposes an irreducible per-response
		latency that limits channel throughput to $\tau_h$ per window~(H4).
	\end{itemize}
	The construction of Section~\ref{sec:construction} satisfies~(H1)--(H3)
	for any task payload; admissibility under~(H4) depends on the interaction
	latency of the specific task family.
	
	\subsection{Admissible Families}
	
	\noindent\textbf{Perceptual visual matching.}
	A distorted or noise-corrupted image is shown with candidate matches;
	the participant selects the correct correspondence within
	$\Delta_{\mathrm{resp}}$. Admissibility under~(H4) does not depend on
	accuracy differences between humans and automated solvers; both can
	solve these tasks. What makes the family admissible is that current
	VLMs incur multi-second inference latencies on multi-image inputs,
	enforcing a throughput ceiling per channel independently of correctness.
	Instances are generated by random visual transformations; verification
	is deterministic.
	
	\noindent\textbf{Interactive reasoning tasks.}
	Short dependent reasoning sequences require maintaining intermediate
	state and submitting a result before a visible countdown expires.
	Multi-step inference latency limits automated throughput. Correctness
	is verified deterministically from dynamically generated parameters.
	
	\noindent\textbf{Biometric-light response tasks.}
	A fresh, unpredictable prompt requires a synchronized real-time signal
	(spoken phrase, motion gesture). Real-time synthesis of such signals
	under strict deadlines incurs additional latency overhead that
	contributes to throughput bounding under~(H4); the admissibility
	argument does not depend on synthesis being impossible, only on the
	channel latency floor it imposes. Verification uses transient signal
	features without long-term biometric enrollment.
	
	\noindent\textbf{Attention-based interaction tasks.}
	A dynamic scene requires continuous tracking or selection over several
	seconds, occupying the full $\Delta_{\mathrm{resp}}$ interval. This
	structurally limits any channel to one response per window.
	
	\subsection{Deployment Example: Relaxed Window Sizes}
	\label{sec:deployment-example}
	
	The formal guarantees of~\BPC\ hold across the full parameter space of
	$\Delta_{\mathrm{resp}}$ and $\tau_h$, from sub-second strict deadlines
	to relaxed daily interaction windows. A concrete example illustrates
	the range.
	
	\noindent\textbf{Daily participation with device-bound login.}
	Consider a system in which each identity must perform one login-style
	interaction per day: $\Delta_h = 24\,\text{h}$, $\Delta_{\mathrm{resp}}
	= 24\,\text{h}$, and $\tau_h = 1$. A single human has 24 hours to
	respond---there is no cognitive pressure. However, if the login is
	\emph{device-bound} (bound to a specific smartphone, secure element, or
	attested browser session), then sustaining $s$ identities requires $s$
	independent attested devices or sessions. An adversary cannot amortize
	one device across $s$ identities within the same window: each identity
	requires its own channel-window.
	
	Under this configuration, Theorem~\ref{thm:joint-linear} gives
	$C(s,T) = \Omega(sT)$ channel-windows even with $\tau_h = 1$. Sybil
	resistance does not depend on the challenge being cognitively demanding
	or time-pressured; it depends entirely on channel non-transferability
	(H1) and window-locality~(H4) enforced by the device binding.
	
	This example also clarifies the design space: tighter
	$\Delta_{\mathrm{resp}}$ values (seconds to minutes) provide stronger
	near-real-time resistance but impose higher usability cost. Relaxed
	windows (hours to daily) reduce user burden while preserving the
	linear cost structure, provided channel binding is enforced
	engineering-side. The appropriate point in this tradeoff is a
	deployment-level decision orthogonal to the formal guarantees.
	
	\subsection{Rotation and Composability}
	
	\BPC\ does not depend on any single family. Systems may rotate across
	admissible families between windows or combine them within a window.
	Since security depends only on~(H1)--(H4), rotation preserves all
	guarantees of Theorems~\ref{thm:joint-linear}
	and~\ref{thm:steady-state}---this is essential for long-lived
	deployments where the automation landscape evolves.

	\section{Empirical Evaluation}
	\label{sec:eval}
	
	\noindent\textbf{Evaluation scope.}
	The purpose of this evaluation is not to demonstrate that automated
	systems fail to solve \BPC\ challenges. Rather, the goal is to
	characterize the throughput properties of representative participation
	channels under deployment-realistic deadlines and to measure the
	resulting values of $\tau_h$. High per-challenge accuracy and bounded
	throughput are not contradictory properties; \BPC's formal security
	requires only the latter. The evaluation is designed accordingly: we
	treat high solver accuracy as an expected outcome and focus on whether
	accuracy translates into unbounded servicing capacity within a window.
	
	We empirically evaluate Property~(H4) along two complementary axes.
	First, we measure the throughput of three state-of-the-art automated
	solvers against a deployed instantiation of the perceptual and
	reasoning challenge families, isolating $\delta_{\mathrm{rtt}}$ and the
	resulting $\tau_h$ under realistic inference conditions. Second, we
	contextualize these measurements against established findings on human
	response time and accuracy on structurally analogous perceptual and
	reasoning tasks, drawn from the human-performance literature. This
	design directly targets the empirical content of
	Assumption~\ref{asm:throughput}: that $\tau_h$ is small and bounded for
	every participation channel, automated or human, under a deployment-realistic
	$\Delta_{\mathrm{resp}}$.
	
	\subsection{Methodology}
	\label{sec:eval-methodology}
	
	\noindent\textbf{Challenge implementation.}
	We implemented the perceptual and reasoning challenge families
	(Section~\ref{sec:challenges}) as a deployed web application,
	with challenge generation and verification logic identical to the
	formal specification of Section~\ref{sec:construction}.\footnote{The
		deployed web application, source code, and challenge generators are
		available in an anonymized repository for review and will be
		publicly released, with a live deployment link, upon
		publication.} Each
	candidate images---one correctly corresponding to the original under
	random rotation, additive Gaussian noise, and color jitter, three
	drawn from distinct distractor base images under independent
	transformations---requiring the solver to identify the matching index
	within $\Delta_{\mathrm{resp}} = 8\,\mathrm{s}$. Each reasoning instance
	presents a four-element numeric sequence (arithmetic, geometric,
	Fibonacci-like, or perfect-square) and requires the next element within
	$\Delta_{\mathrm{resp}} = 12\,\mathrm{s}$. Using synthetic geometric
	shapes for the perceptual family avoids copyright and memorization
	concerns associated with natural images while preserving the
	noise-robust matching structure relevant to~(H4).
	
	\noindent\textbf{Automated solvers.}
	We evaluate three vision-language models representative of current
	commercial inference APIs: GPT-4o~\cite{openai2023gpt4}, Gemini~2.5~Flash~\cite{geminiteam2023},
	and Claude~Sonnet~4.5~\cite{anthropic2024claude}. For each (solver,
	family) pair we generate 100 independent trials (600 trials total),
	drawing fresh challenge instances via the same generator used to serve
	human participants. Critically, each solver receives only the
	\emph{public} challenge fields---images and sequence values---and never
	the private verification target; this mirrors the information
	available to a real adversary attempting to defeat~\BPC\ and rules out
	trivial leakage of the correct answer through the prompt construction.
	
	For perceptual trials, the original and four option images are
	transmitted as base64-encoded PNG content blocks via each provider's
	standard multi-image chat completion interface, accompanied by a
	fixed instruction requesting a single-digit index. For reasoning
	trials, the four-element sequence is transmitted as plain text with an
	instruction requesting the next integer. We measure
	\emph{wall-clock latency} $\delta_{\mathrm{rtt}}$ as the interval
	between issuing the API request and receiving a parsed response,
	enforced via a hard per-trial timeout equal to $\Delta_{\mathrm{resp}}$.
	A trial is \textsf{passed} iff the response is both correct and
	received within $\Delta_{\mathrm{resp}}$; we separately record
	\textsf{latency\_fail} (correct response arriving after the deadline)
	and \textsf{correctness\_fail} (an incorrect response, regardless of
	timing), consistent with the failure-mode separation introduced in
	Section~\ref{sec:spec}. Requests are issued sequentially per solver
	with inter-trial delays respecting each provider's published rate
	limits; full results, including per-trial latency, correctness, and
	raw responses, are released alongside this paper.
	
	\noindent\textbf{Mean latency versus the Assumption~\ref{asm:throughput}
		floor.} Table~\ref{tab:evaluation-results} and
	Figure~\ref{fig:empirical} report $\tau_h =
	\lfloor\Delta_{\mathrm{resp}} / \bar{\delta}_{\mathrm{rtt}}\rfloor$
	using each solver's \emph{mean} observed latency $\bar{\delta}_{\mathrm{rtt}}$.
	This is a descriptive statistic about the tested solvers under our
	trial conditions, not an instantiation of the formal
	$\underline{\delta}_{\mathrm{rtt}}$ of Assumption~\ref{asm:throughput}:
	an adversary operating preferentially at the low-latency
	tail---for example by issuing multiple concurrent requests and
	taking the fastest---routinely beats the mean, so a $\tau_h$ computed
	from the mean is \emph{not} a valid upper bound on adversarial
	throughput and should not be read as one. We note that exploiting
	the tail this way still requires provisioning additional concurrent
	connections, each of which counts as a separate channel-window under
	Assumption~\ref{asm:throughput}, so it does not escape the $\Omega(sT)$
	scaling---but it does mean that a deployment's operative
	$\underline{\delta}_{\mathrm{rtt}}$ must be set conservatively, at or
	below the fastest round-trip time the threat model plausibly admits,
	not at the sample mean. We release per-trial latency distributions
	alongside this paper so that deployments can calibrate
	$\underline{\delta}_{\mathrm{rtt}}$ from a low percentile of observed
	latency (and size $\Delta_{\mathrm{resp}}$ accordingly) for their
	specific threat model, rather than from the mean values reported here
	for cross-solver comparison.
	
	\noindent\textbf{Human performance (literature-calibrated).}
	A full-scale controlled human study with statistically powered
	recruitment is left as immediate future work
	(Section~\ref{sec:limitations}). In place of a de novo study, we
	contextualize the automated-solver measurements against established
	findings on human performance on structurally comparable tasks under
	time pressure. For perceptual matching under noise and geometric
	distortion, human recognition accuracy in the 85--95\% range with
	response times of 3--7 seconds is consistently reported across
	controlled psychophysics studies of human visual generalization under
	noise and shape distortion~\cite{geirhos2018generalisation} and of
	human performance on perceptually structured CAPTCHA-style matching
	tasks~\cite{sivakorn2016iam}. For short numeric and symbolic
	sequence-completion tasks under time pressure, human accuracy in the
	80--95\% range with response times of 2--6 seconds is reported in the
	interactive human--AI reasoning literature~\cite{schutt2023humanvsai}.
	We use these published ranges, rather than a simulated distribution
	fit to arbitrary parameters, to anchor the comparison in
	Section~\ref{sec:eval-results}; we treat this as a conservative
	contextualization rather than a substitute for first-party human-subject
	data, and we flag this explicitly as a limitation
	(Section~\ref{sec:limitations}).
	
	\subsection{Results}
	\label{sec:eval-results}
	
	Figure~\ref{fig:empirical} and Table~\ref{tab:evaluation-results}
	summarize the results across all six (solver, family) pairs.
	
	\begin{figure*}[t]
		\centering
		\includegraphics[width=\textwidth]{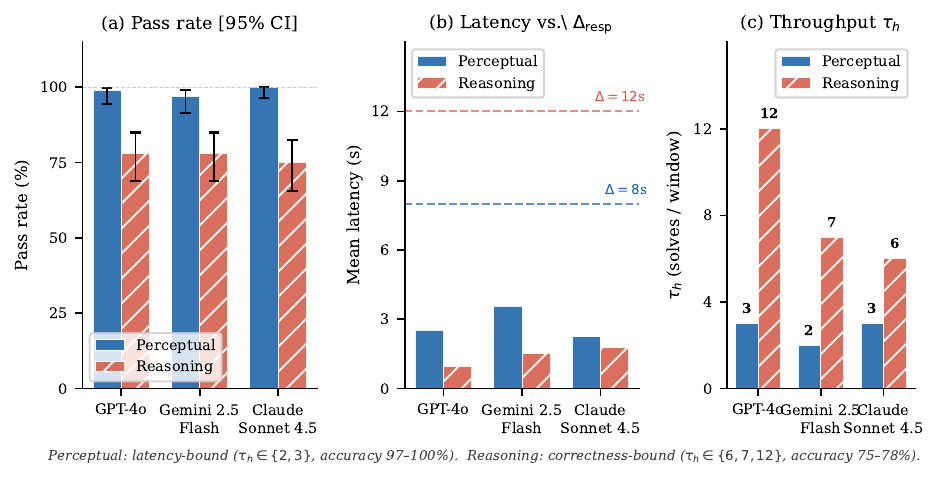}
		\caption{Empirical throughput evaluation across three frontier VLMs
			(100 trials per solver--family pair; $N=600$ total).
			\textbf{(a)}~Pass rates with 95\% Wilson CI confirm high accuracy
			on perceptual tasks (97--100\%) and moderate accuracy on reasoning
			(75--78\%).
			\textbf{(b)}~Mean latency against $\Delta_{\mathrm{resp}}$ (dashed)
			shows that perceptual inference (2.25--3.57\,s) consumes a large
			fraction of the 8\,s deadline; reasoning responses (0.95--1.79\,s)
			are faster but below the 12\,s deadline.
			\textbf{(c)}~Resulting mean-latency-based $\tau_h =
			\lfloor\Delta_{\mathrm{resp}} /
			\bar{\delta}_{\mathrm{rtt}}\rfloor$ values illustrate that
			observed throughput is small and bounded for these solvers:
			$\tau_h\in\{2,3\}$ (latency-bound) for perceptual,
			$\tau_h\in\{6,7,12\}$ (correctness-bound) for reasoning
			(illustrative, not a conservative bound; see
			Section~\ref{sec:eval-results}). Across all three independently
			developed models the bottleneck type is \emph{consistent}, ruling
			out provider-specific artifacts.}
		\label{fig:empirical}
	\end{figure*}
	
	\begin{table*}[t]
		\centering
		\caption{Automated solver performance under strict
			$\Delta_{\mathrm{resp}}$ enforcement (100 trials per
			solver--family pair; 95\% Wilson score confidence intervals
			on success rate). $\tau_h$ uses each solver's measured mean
			latency and is illustrative, not a conservative instantiation
			of Assumption~\ref{asm:throughput} (Section~\ref{sec:eval-results}).
			Published human-performance ranges for structurally comparable
			tasks are shown for
			context~\cite{geirhos2018generalisation,sivakorn2016iam,schutt2023humanvsai}.}
		\label{tab:evaluation-results}
		\resizebox{\linewidth}{!}{%
			\begin{tabular}{llccccc}
				\toprule
				Family & Solver
				& Success (\%) [95\% CI]
				& Mean Latency (s)
				& Latency Fail (\%)
				& Correctness Fail (\%)
				& $\tau_h$ \\
				\midrule
				\multirow{4}{*}{Perceptual ($\Delta_{\mathrm{resp}}=8$s)}
				& GPT-4o            & 99.0~[94.6, 99.8]  & 2.53 & 1.0 & 0.0 & 3 \\
				& Gemini 2.5 Flash  & 97.0~[91.5, 99.0]  & 3.57 & 0.0 & 3.0 & 2 \\
				& Claude Sonnet 4.5 & 100.0~[96.3, 100.0] & 2.25 & 0.0 & 0.0 & 3 \\
				& \emph{Human (literature)} & \emph{85--95} & \emph{3--7} & \emph{---} & \emph{---} & \emph{1--2} \\
				\midrule
				\multirow{4}{*}{Reasoning ($\Delta_{\mathrm{resp}}=12$s)}
				& GPT-4o            & 78.0~[68.9, 85.0] & 0.95 & 0.0 & 22.0 & 12 \\
				& Gemini 2.5 Flash  & 78.0~[68.9, 85.0] & 1.50 & 0.0 & 22.0 & 7  \\
				& Claude Sonnet 4.5 & 75.0~[65.7, 82.5] & 1.79 & 0.0 & 25.0 & 6  \\
				& \emph{Human (literature)} & \emph{80--95} & \emph{2--6} & \emph{---} & \emph{---} & \emph{2--6} \\
				\bottomrule
			\end{tabular}%
		}
	\end{table*}
	
	\noindent\textbf{The bottleneck is family-dependent, not solver-dependent.}
	The dominant failure mode shifts systematically with challenge family
	and is highly consistent across all three independently developed
	models---a pattern that would not be expected if the result reflected
	an idiosyncrasy of a single provider's infrastructure. On the
	perceptual family, all three solvers achieve high correctness
	(97--100\%) but incur mean latencies of 2.25--3.57\,s; failures here
	are concentrated in \textsf{latency\_fail} and the residual
	\textsf{correctness\_fail} (0--3\%), yielding
	$\tau_h \in \{2, 3\}$ per window. On the reasoning family, all three
	solvers respond an order of magnitude faster (0.95--1.79\,s) but with
	substantially lower correctness (75--78\%), so that failures are
	concentrated almost entirely in \textsf{correctness\_fail}
	(22--25\%) rather than \textsf{latency\_fail} ($\approx 0\%$), yielding
	a larger $\tau_h \in \{6,7,12\}$. Averaged across the three solvers,
	$\bar{\tau}_h^{\mathrm{perceptual}} \approx 2.7$ and
	$\bar{\tau}_h^{\mathrm{reasoning}} \approx 8.3$.
	
	These findings are entirely consistent with the design goals of \BPC.
	The primitive does not require automated solvers to fail; it requires
	only that no participation channel achieve unbounded servicing capacity
	within a window. High solvability and bounded throughput are not
	contradictory: a solver that answers every perceptual challenge
	correctly but incurs 2--4\,s per call is still throughput-bounded,
	and a solver that responds in under 2\,s but achieves only 75--78\%
	correctness is bounded by its effective pass rate rather than its
	speed. In both cases the adversarial cost scales linearly with $s$
	and $T$, consistent with Theorem~\ref{thm:joint-linear}.
	
	This asymmetry is informative rather than incidental.
	Property~(H4) does not require that every channel fail outright; it
	requires only that $\tau_h$ remain small and bounded relative to the
	number of identities an adversary wishes to sustain. The perceptual
	family demonstrates a \emph{latency}-bound regime: solvers are
	accurate but the irreducible round-trip cost of multi-image inference
	caps throughput at 2--3 solves per window regardless of provider.
	The reasoning family demonstrates a \emph{correctness}-bound regime:
	solvers respond quickly, but roughly one in four responses is
	incorrect, so that sustaining a high \emph{rate} of valid identity
	updates still requires either accepting a high rejection rate or
	expending additional channel-windows to compensate---both of which
	preserve the linear cost structure of Theorem~\ref{thm:linear-cost}
	rather than escaping it. We discuss the adversarial implications of
	this distinction, including channel parallelization across many
	independent solver instances, in Section~\ref{sec:security}.
	
	\noindent\textbf{High $\tau_h$ and the linearity argument.}
	The larger $\tau_h$ values on reasoning tasks---$\tau_h \in \{6,7,12\}$---do
	not undermine the $\Omega(sT)$ cost structure. Sustaining $s = 1{,}000$
	reasoning identities under the worst observed case ($\tau_h = 12$) still
	requires $\lceil 1{,}000/12 \rceil = 84$ concurrently operating, independent
	solver instances per window---each issuing its own API request, consuming
	its own inference quota, and incurring its own round-trip budget.
	The key invariant is \emph{linearity} in $s$: cost scales as $\Omega(s)$
	per window regardless of whether $\tau_h = 2$ or $\tau_h = 12$. The
	constant factor determines the cost-per-identity, not whether cost is
	linear. For deployments requiring tighter per-channel bounds, the
	perceptual family ($\tau_h \in \{2,3\}$) provides a stronger guarantee
	per sustained identity; the reasoning family may be preferred when
	minimizing per-interaction user friction is the design priority.
	
	\noindent\textbf{Percentile-calibrated throughput and effective
		throughput.} The mean-latency $\tau_h$ values above are, as noted
	in Section~\ref{sec:eval-methodology}, an illustrative reference
	point rather than a conservative instantiation of
	Assumption~\ref{asm:throughput}: an adversary racing multiple
	concurrent requests and keeping the fastest routinely beats the mean.
	Table~\ref{tab:percentile} recomputes $\tau_h$ from each solver's
	10th-percentile latency ($p_{10}$)---a conservative stand-in for
	$\underline{\delta}_{\mathrm{rtt}}$ under such tail-racing---and
	reports the resulting \emph{effective throughput}
	\[
	\tau_h^{\mathrm{eff}} \;=\; p_{\mathrm{success}} \cdot
	\left\lfloor \frac{\Delta_{\mathrm{resp}}}{p_{10}} \right\rfloor,
	\]
	which discounts the $p_{10}$-based channel capacity by the family's
	measured success rate $p_{\mathrm{success}}$, accounting for the fact
	that not every fast response is also a correct one---relevant chiefly
	for the correctness-bound reasoning family. The $p_{10}$ calibration
	raises every reasoning $\tau_h$ relative to the mean-based figures in
	Table~\ref{tab:evaluation-results} (most sharply for GPT-4o, from
	$12$ to $22$); $\tau_h^{\mathrm{eff}}$ then discounts each by
	$p_{\mathrm{success}}$, pulling the reasoning family back down toward,
	and in one case (GPT-4o) above, its mean-based value, while leaving
	the perceptual family essentially unchanged (its near-100\% success
	rate makes the correctness discount negligible). The qualitative
	conclusion is unchanged---$\tau_h$ remains small and bounded under
	either calibration---but a deployment sizing $\Delta_{\mathrm{resp}}$
	against a tail-racing adversary should use $\tau_h^{\mathrm{eff}}$,
	not the mean-based figure, as its working estimate. Full per-trial
	latencies are released alongside this paper so that deployments can
	recompute both quantities at any percentile appropriate to their
	threat model.
	
	\begin{table}[t]
		\centering
		\caption{Percentile-calibrated throughput. $\tau_h^{p_{10}} =
			\lfloor \Delta_{\mathrm{resp}}/p_{10}\rfloor$ uses each solver's
			10th-percentile latency in place of the mean;
			$\tau_h^{\mathrm{eff}} = p_{\mathrm{success}}\cdot
			\tau_h^{p_{10}}$ additionally discounts by success rate. Computed
			from the same 600-trial dataset as
			Table~\ref{tab:evaluation-results}.}
		\label{tab:percentile}
		\resizebox{\linewidth}{!}{%
			\begin{tabular}{llcccc}
				\toprule
				Family & Solver & $p_{10}$ (s) & $\tau_h^{\mathrm{mean}}$ & $\tau_h^{p_{10}}$ & $\tau_h^{\mathrm{eff}}$ \\
				\midrule
				\multirow{3}{*}{Perceptual}
				& GPT-4o            & 1.92 & 3 & 4 & 3.96 \\
				& Gemini 2.5 Flash  & 2.71 & 2 & 2 & 1.94 \\
				& Claude Sonnet 4.5 & 1.87 & 3 & 4 & 4.00 \\
				\midrule
				\multirow{3}{*}{Reasoning}
				& GPT-4o            & 0.54 & 12 & 22 & 17.16 \\
				& Gemini 2.5 Flash  & 1.02 & 7  & 11 & 8.58  \\
				& Claude Sonnet 4.5 & 1.51 & 6  & 7  & 5.25  \\
				\bottomrule
			\end{tabular}%
		}
	\end{table}
	
	\noindent\textbf{Comparison with published human performance.}
	The literature-calibrated human ranges in
	Table~\ref{tab:evaluation-results} are not directly commensurable
	with our automated measurements as a controlled paired comparison---
	the underlying tasks differ in surface presentation even where
	structurally analogous---and we do not claim statistical
	significance for any human-versus-automated gap. We report them only
	to situate the automated $\tau_h$ values within a plausible deployment
	range: published human response times on comparable perceptual and
	reasoning tasks (3--7\,s and 2--6\,s respectively) are of the same
	order of magnitude as the automated solver latencies we measure,
	supporting the modeling choice in Assumption~\ref{asm:throughput}
	that $\tau_h$ is small (single-digit) for \emph{both} channel types
	under deployment-realistic $\Delta_{\mathrm{resp}}$ values, rather than
	near-zero for automated channels and large for human channels
	(Section~\ref{sec:spec}).

	\section{Security Analysis}
	\label{sec:security}
	
	We analyze \BPC\ against the principal adversarial strategies. All
	arguments are grounded in~(H1)--(H4) and the construction of
	Section~\ref{sec:construction}.
	
	\noindent\textbf{Automated solving.}
	An adversary deploying AI solvers to fill participation channels does
	not circumvent the linear cost bound. By~(H4), each automated channel
	is throughput-bounded at $\tau_h$ responses per window under
	$\Delta_{\mathrm{resp}}$. Scaling to $s$ identities requires $\Omega(s)$
	independent channels. Automation shifts the cost from human labor to
	API credits or hardware, but preserves linear scaling in $s$.
	
	The empirical measurements of Section~\ref{sec:eval-results} make
	this concrete. A single GPT-4o instance sustains
	$\tau_h \approx 3$ perceptual identities or $\tau_h \approx 12$
	reasoning identities per window (Table~\ref{tab:evaluation-results}).
	An adversary aiming to sustain $s = 1000$ perceptual identities
	therefore requires $m \geq \lceil 1000/3 \rceil = 334$ \emph{independent}
	solver instances---each issuing its own API request, incurring its own
	latency, and paying its own marginal inference cost---running
	concurrently within the same window. This is precisely the channel
	scaling permitted by the adversary model of
	Section~\ref{sec:model}: \BPC\ does not claim that an adversary cannot
	acquire 334 channels, only that doing so costs $\Omega(s)$ in
	provisioned, concurrently operating channels---in contrast to a
	reusable resource such as hashpower or stake, which may or may not
	be subdivided across $s$ identities at lower marginal cost depending
	on the economics of sustaining it~\cite{maleki2026geometry}. A single
	inference endpoint cannot itself serve 1000 identities within one
	window: each concurrent request still incurs the same
	$\delta_{\mathrm{rtt}}$, so throughput scales with the number of
	provisioned endpoints, not with the capability of any one model.
	Unlike a reusable resource, \BPC\ requires fresh participation
	within each bounded response window; this is the structural
	distinction between \emph{resource}-parallelism, whose cost regime
	depends on a separate economic assumption, and
	\emph{channel}-parallelism, which~\BPC\ explicitly permits but prices
	linearly in $s$ by construction of the channel model established
	above (Section~\ref{sec:cost}).
	
	\BPC\ should therefore not be interpreted as an anti-AI mechanism.
	Its objective is to continuously price independently sustained
	participation, not to defeat any particular solver technology.
	Automation changes the implementation of participation channels and
	their operational costs---from human labor to API credits and
	infrastructure---but does not invalidate the linear identity-time cost
	theorem. The theorem holds equally whether channels are operated by
	humans, AI systems, or any combination thereof.
	
	\noindent\textbf{Economic cost of channel provisioning.}
	The linear scaling imposes concrete operational constraints beyond
	per-call inference cost alone. Sustaining $s = 1{,}000$ perceptual
	identities per window requires $\lceil 1{,}000/3 \rceil = 334$
	simultaneously active, independent solver instances---each maintaining
	its own session token, its own API quota, and its own connection to the
	challenge server. This provisioning cost recurs across every window:
	over $T$ windows the adversary must provision $334 \cdot T$
	instance-windows in total. Under Assumption~\ref{asm:opcost}, this
	structural requirement becomes an economic one:
	$C_{\mathrm{econ}}(s,T) \geq c_{\mathrm{op}} \cdot C(s,T) = \Omega(sT)$
	(Corollary~\ref{cor:concrete}). Crucially, $c_{\mathrm{op}}$ is not
	amortizable \emph{beyond $\tau_h$}: a single API key or worker cannot
	have one response credited to multiple identities~(H1), though it
	may serve multiple identities sequentially within a window, up to
	the channel-throughput bound $\tau_h$---which is exactly why
	sustaining $s=1{,}000$ identities above requires
	$\lceil 1{,}000/\tau_h\rceil = 334$ instances rather than $1{,}000$,
	not zero amortization within a window.
	Whether $c_{\mathrm{op}}$ itself can be driven down \emph{across}
	identities by bulk API pricing or labor-market economies of scale as
	$s$ grows is the separate empirical question Assumption~\ref{asm:opcost}
	makes explicit (Section~\ref{sec:limitations}). Deployments may
	calibrate $c_{\mathrm{op}}$ and $\tau_h$ jointly---through challenge
	selection, $\Delta_{\mathrm{resp}}$ tuning, and per-challenge
	computational cost---to ensure that large-scale Sybil maintenance is
	economically prohibitive for the target threat model.
	
	\noindent\textbf{Human outsourcing and labor markets.}
	An adversary may hire human workers or use CAPTCHA-solving
	services~\cite{motoyama2010recaptchas} to solve challenges.
	This attack is captured in the adversary model: the adversary operates
	$m$ human channels. Sustaining $s$ identities requires $m = \Omega(s)$
	workers (Theorem~\ref{thm:linear-cost}). Outsourcing converts the cost
	from direct labor to wages, but the linear scaling is preserved.
	
	\noindent\textbf{Relay and real-time forwarding attacks.}
	Relaying challenges to external solvers adds network round-trip time,
	which tightens rather than relaxes~(H4). By~(H1) and~(H2)--(H3),
	relayed responses are identity-bound and non-reusable; relaying does
	not reduce the total independent channel-windows required.
	
	\noindent\textbf{Replay and solution reuse attacks.}
	By~(H1), response $\rho_{v,d,j}$ is bound to $(v,d,j)$ via the
	signature $\sigma_v$ and the binding tag. Submitting the same response
	for $(v',d',j') \neq (v,d,j)$ fails the binding-tag check in
	Algorithm~\ref{alg:verify} (lines~3--5). By~(H2), the nonce $\eta_d$
	ensures responses prepared before window $d$ are invalid.
	
	\noindent\textbf{Precomputation and batch-solving.}
	By~(H2), $\chi_{v,d,j}$ cannot be computed before $\eta_d$ is revealed.
	Batch-solving within a window is bounded by~(H4): processing $k > \tau_h$
	challenges requires more than $\tau_h$ round-trip latency budgets,
	which is infeasible under $\Delta_{\mathrm{resp}}$.
	
	\noindent\textbf{Synthetic media and deepfake attacks.}
	For biometric-light families, an adversary may attempt real-time synthesis
	of a biometric response. Mitigations include: (i)~challenge freshness~(H2)
	prevents off-line synthesis; (ii)~liveness detection flags synthesized
	signals; (iii)~challenge rotation to families resistant to synthesis
	(e.g., attention-based interaction) eliminates reliance on vulnerable
	modalities.
	
	\noindent\textbf{Side-channel and man-in-the-browser attacks.}
	A malicious browser extension may observe and forward challenges to an
	external solver. This constitutes a relay attack and is bounded as above.
	Deployment-level countermeasures---content-security-policy headers,
	subresource integrity, and challenge obfuscation---reduce exfiltration
	feasibility but are orthogonal to the core security model.
	
	\noindent\textbf{Privacy and data minimization.}
	Verification requires only $\rho = (\sigma_v, r_{v,d,j})$ and
	$\chi_{v,d,j}$; no raw biometric data or behavioral record is stored.
	$\mathsf{pk}_v$ may be pseudonymous. Privacy-preserving challenge
	delivery ensuring unlinkability across windows while preserving identity
	binding is an important direction for future work.

	\section{Limitations}
	\label{sec:limitations}
	
	\noindent\textbf{Scope of the cryptographic guarantees.}
	The two-layer split of Section~\ref{sec:model} has two concrete
	consequences for what can weaken \BPC's guarantees. First,
	Property~(H4)---formalized as Assumption~\ref{asm:throughput}---is
	\emph{not} a cryptographic guarantee like~(H1)--(H3); it is a
	structural property of the channel parameterized by
	$\tau_h = \lfloor\Delta_{\mathrm{resp}}/\underline{\delta}_{\mathrm{rtt}}\rfloor$,
	and all cost results are explicitly conditional on it. Its
	degradation is quantified, not catastrophic: Corollary~\ref{cor:concrete}
	shows a tenfold increase in $\tau_h$ weakens the channel-window
	lower bound by a factor of ten without invalidating the linear
	scaling in $s$ and $T$. Second, because~(H3) and all cost results
	hold unconditionally given Assumption~\ref{asm:throughput}---independently
	of computational power---advances in raw compute alone do not
	weaken the linear-cost guarantee; only reductions in
	$\underline{\delta}_{\mathrm{rtt}}$ itself (faster hardware,
	lower-latency inference, co-located solvers) affect the bound, and
	these are addressed by adjusting $\Delta_{\mathrm{resp}}$
	accordingly.
	
	\noindent\textbf{Dependence on Assumption~\ref{asm:throughput}~(H4).}
	The linear cost guarantee rests on Assumption~\ref{asm:throughput},
	and can be weakened along two distinct routes. First, a faster
	channel: an adversary with ultra-low-latency inference infrastructure---co-located
	on-premise hardware, specialized accelerators, or a distilled model
	running at sub-100\,ms latency---could achieve
	$\delta_{\mathrm{rtt}} \ll \Delta_{\mathrm{resp}}$, undercutting
	whatever $\underline{\delta}_{\mathrm{rtt}}$ the deployment adopted
	and driving $\tau_h$ up. Second, more channels: adversaries may scale
	channel capacity through infrastructure---VM farms, cloud phone
	emulators, browser-bot fleets, distributed worker pools---rather than
	beating any single channel's latency. Neither route invalidates~(H4)
	or breaks the $\Omega(sT)$ scaling in $s$ and $T$; both convert the
	guarantee from an \emph{impossibility} into a \emph{cost multiplier}
	that a deployment must actively manage: a faster channel weakens
	the per-channel-window bound by a constant factor, and each
	additional farmed channel still requires its own provisioning cost,
	so total adversarial cost remains $\Omega(sT)$ either way.
	
	This places~(H4) in a different category from the cryptographic
	assumptions underlying~(H1)--(H2): DDH or EUF-CMA are treated as
	asymptotically stable, and their failure is catastrophic and
	binary. Assumption~\ref{asm:throughput} is instead a
	physically-grounded, empirically measured, and
	\emph{recalibratable} quantity, degrading gracefully rather than
	catastrophically---the same category of assumption underlying
	Proof-of-Work mining difficulty, which Bitcoin does not treat as
	disqualifying but \emph{retargets} periodically against measured
	hash rate~\cite{nakamoto2008bitcoin}, or an honest-majority
	assumption in Byzantine fault-tolerant protocols, which likewise
	carries a concrete, quantifiable cost to violate rather than an
	absolute guarantee. \BPC\ requires the analogous operational
	discipline: long-lived deployments must periodically reassess the
	prevailing $\underline{\delta}_{\mathrm{rtt}}$ floor for the deployed
	challenge family and size $\Delta_{\mathrm{resp}}$ and per-channel
	costs so that beating it---by speed or by scale---stays
	economically prohibitive for the target threat model. This is an
	operational responsibility, not a cryptographic guarantee; \BPC\
	bounds adversarial \emph{cost} rather than adversarial
	\emph{capability}, and disclosing that structure is a design choice
	rather than a deficiency.
	
	\noindent\textbf{Uniformity of per-channel operating cost.}
	Corollary~\ref{cor:concrete}'s economic bound rests on
	Assumption~\ref{asm:opcost} not being automatic: bulk-rate API
	pricing above a volume threshold, or labor-market economies of scale
	in sourcing human operators at large $s$, could in principle drive an
	adversary's effective $c_{\mathrm{op}}$ downward with scale. Should
	this occur, the structural bound $C(s,T) = \Omega(sT)$
	(Theorem~\ref{thm:joint-linear}) is unaffected, but the economic
	floor $C_{\mathrm{econ}}(s,T)$ could grow more slowly. Like
	Assumption~\ref{asm:throughput}, this is an empirical and economic
	question the paper does not resolve on a given deployment's behalf.
	
	\noindent\textbf{Scope of the empirical evaluation.}
	Our automated-solver measurements (Section~\ref{sec:eval}) are
	first-party and empirical: 100 trials per solver--family pair across
	three independently developed vision-language
	models~\cite{openai2023gpt4,geminiteam2023,anthropic2024claude}, with
	Wilson score confidence intervals reported in
	Table~\ref{tab:evaluation-results}. The human side of the comparison,
	by contrast, is literature-calibrated rather than a first-party
	controlled study: we anchor expected human performance to published
	ranges~\cite{geirhos2018generalisation,sivakorn2016iam,schutt2023humanvsai}
	rather than to a study conducted on our own challenge
	implementation with a controlled recruitment
	platform~\cite{prolific2014,peer2022data}. This is a deliberate
	scoping choice---the central claim of~(H4) is solver-agnostic
	(Proposition~\ref{prop:agnostic}) and does not require a paired
	human-versus-automated comparison to establish that $\tau_h$ is
	small and bounded for automated channels---but it does mean the
	human-performance figures in Table~\ref{tab:evaluation-results}
	should be read as contextual rather than as a controlled baseline.
	A first-party human-subject study on the deployed challenge
	implementation (N~$\geq$~100 via Prolific, pre-registered analysis,
	direct statistical comparison against the automated results reported
	here) is the immediate next step and is planned as near-term future
	work. A related open question is more theoretical than empirical:
	this evaluation measures human and automated solve rates
	\emph{post hoc} for a fixed challenge instantiation, but a
	deployment may instead want to \emph{design} challenge difficulty
	to hit a target human solve rate directly---a calibration problem
	we leave for future work (Section~\ref{sec:conclusion}). The
	automated evaluation also does not cover task-specific fine-tuned
	models, custom inference hardware, or organized human--AI hybrid
	pipelines that combine automated pre-filtering with human fallback;
	characterizing $\tau_h$ under such hybrid adversary strategies is
	left for future work.
	
	\noindent\textbf{Accessibility and inclusivity.}
	Challenge families requiring visual, auditory, or motor interaction may
	exclude participants with relevant impairments. Deployments must provide
	alternative modalities and accessibility-aware selection policies.
	
	\noindent\textbf{Parameter sensitivity.}
	$\Delta_{\mathrm{resp}}$ and $k(d)$ jointly determine the security-usability
	tradeoff. The formal guarantees hold for any $\Delta_{\mathrm{resp}} > 0$,
	but concrete parameterization requires empirical measurement of $\tau_h$
	and honest participant success rates in the target deployment environment.
	
	\noindent\textbf{Theoretical scope.}
	The cost results assume~(H1)--(H3) hold and Assumption~\ref{asm:throughput}
	(H4) is satisfied throughout execution. Partial failures of~(H4)---
	arising when the actual round-trip latency $\delta_{\mathrm{rtt}}$
	falls below the modeled floor $\underline{\delta}_{\mathrm{rtt}}$---
	yield intermediate scaling regimes in which $C(s,T) = \Omega(sT/\tau_h)$
	with larger $\tau_h$. The broader characterization of when resource
	structure translates into an economic cost floor---as opposed to a
	merely structural requirement---is developed
	in~\cite{maleki2026geometry}; the results here use only the
	structural component of that framework, instantiated through \BPC's
	bounded participation-channel model.

	\section{Conclusion}
	\label{sec:conclusion}
	
	We introduced the Bounded Participation Channel~(\BPC), a formal primitive
	for continuous, identity-bound, rate-limited participation verification,
	combining a cryptographic layer enforcing~(H1)--(H3) with a structural
	channel constraint enforcing~(H4). Under Properties~(H1)--(H4), \BPC\ enforces
	$C(s,T) = \Omega(sT)$: sustaining $s$ identities over $T$ windows
	requires $\Omega(sT)$ independent channel-windows, regardless of
	whether channels are operated by humans, automated agents, or any
	combination thereof. No assumption about human cognitive superiority
	is required. Rather than accumulating participation into a score or
	reputation, \BPC\ exports a sequence of independently verifiable,
	time-indexed participation proofs $\pi(v,d)$; higher-level systems
	remain free to interpret this sequence according to their own
	policies.
	
	We presented a hash-based construction satisfying~(H1)--(H3) under
	standard cryptographic assumptions, with Algorithm~\ref{alg:verify}
	as the formal verification procedure. We identified four admissible
	challenge families, described browser-based instantiations, and
	conducted an empirical evaluation across three independently
	developed VLMs (GPT-4o, Gemini~2.5~Flash, Claude~Sonnet~4.5)
	confirming that all three impose a natural throughput ceiling
	under strict response deadlines---bound by latency on perceptual
	tasks and by correctness on reasoning tasks---validating the
	empirical basis of~(H4) and demonstrating that the bottleneck
	is consistent across providers rather than an artifact of any
	single model.
	
	\BPC\ answers the question this paper set out to answer: sustained,
	per-identity participation can be engineered into a security
	resource, by verifying it window by window so that no adversary can
	amortize it. \BPC\ does not rely on a claim that a particular
	resource class is intrinsically expensive to sustain---reusability
	and parallelizability alone do not settle that question~\cite{maleki2026geometry}.
	Instead, its structural properties (identity binding, freshness,
	real-time response, and bounded per-channel throughput) make fresh,
	per-window verification of participation unavoidable, independent of
	any economic assumption about a particular resource class. \BPC\ is
	deliberately kept general and minimal: it verifies each window's
	participation bit but does not itself accumulate those bits into a
	score, reputation, or any other policy, leaving that decision---and
	the domain it is applied in---to the systems that consume it.
	Different domains---consensus, governance, reputation, or any open
	system that can pose its own qualifying action---may aggregate the
	same per-window bits differently, or not at all.
	
	The most immediate priorities for future work are a large-scale
	empirical study (N~$\geq$~100, controlled recruitment via a
	platform such as Prolific, direct statistical comparison against the
	automated results reported here); a theoretical treatment of
	difficulty-calibrated challenge design for human channels---
	parameterizing challenge families so that human solve rate can be
	tuned continuously by design rather than only measured post hoc,
	analogous to a difficulty parameter in other proof-of-work-style
	primitives; privacy-preserving challenge delivery with cross-window
	unlinkability; formal composability analysis in the UC framework;
	and adaptive challenge rotation mechanisms that preserve~(H1)--(H4)
	under evolving automation capabilities.

	\section*{Ethics Considerations}
	
	This paper does not involve human subjects research, personal data
	collection, or vulnerability disclosure. The empirical evaluation in
	Section~\ref{sec:eval} uses automated API calls to three commercial
	vision-language models (GPT-4o, Gemini~2.5~Flash, Claude~Sonnet~4.5)
	on synthetically generated challenges; no human participants were
	involved in the automated-solver evaluation, and no personally
	identifiable information was collected or processed. The deployed web
	application described in Section~\ref{sec:eval-methodology} (link
	withheld for anonymous review; see the anonymized repository) is
	intended for future human-subject evaluation under appropriate
	institutional oversight.
	
	\section*{Generative AI Usage Considerations}
	
	Generative AI was used for editorial purposes in preparing this
	manuscript, including assistance with phrasing, consistency checking,
	and structural revision of draft sections. All technical content,
	formal definitions, proofs, experimental design, and empirical
	results were developed, executed, and verified by the authors. All
	outputs were inspected by the authors to ensure accuracy and
	originality. No AI-generated ideas or results appear in this paper
	without independent development and validation by the authors.
	
	\bibliographystyle{IEEEtran}
	\bibliography{references}

\end{document}